\documentclass[lettersize,journal]{IEEEtran}
\usepackage[caption=false,font=normalsize,labelfont=sf,textfont=sf]{subfig}
\usepackage{float}
\usepackage{textcomp}
\usepackage{cite}
\usepackage{amsmath,amssymb,amsfonts}
\usepackage{algorithmic}
\usepackage{amsfonts}
\usepackage{graphicx}
\usepackage{textcomp}
\usepackage{xcolor}
\usepackage{amssymb}
\usepackage{epstopdf}
\usepackage{enumitem}
\usepackage{pifont}
\usepackage{algorithmic}
\usepackage{algorithm}
\usepackage{bm}
\usepackage{lipsum}
\usepackage{cuted}
\usepackage{booktabs}
\usepackage[mathscr]{eucal}
\usepackage{stfloats}
\usepackage{url}
\usepackage{upgreek}
\usepackage{enumitem}
\usepackage{amsthm}
\usepackage{stfloats}
\usepackage{arydshln}
\usepackage{tikz}
\newlist{dblparenenum}{enumerate}{1}
\setlist[dblparenenum,1]{label={{(}\arabic*{{)}}}}
\newtheoremstyle{mystyle}
  {0pt}
  {0pt}
  {\normalfont}
  {\parindent}
  {\itshape}
  {:}
  { }
  {\thmname{#1}\thmnumber{ #2}\thmnote{ (#3)}}
\theoremstyle{mystyle}
\newtheorem{Definition}{Definition}
\newtheorem{Corollary}{Corollary}
\newtheorem{remark}{Remark}

\newtheorem{proposition}{Proposition}
\newenvironment{Proof}{{\noindent \textit {Proof:}}}{\hfill $\square$ \par}
\usepackage{hyperref}
\hypersetup{
    colorlinks=true,
    linkcolor=blue,
    citecolor=blue,
    filecolor=magenta,
    urlcolor=blue
}

\def\BibTeX{{\rm B\kern-.05em{\sc i\kern-.025em b}\kern-.08em
    T\kern-.1667em\lower.7ex\hbox{E}\kern-.125emX}}
    \begin{document}

\title{Fundamental Limits of Pulse Based UWB ISAC Systems: A Parameter Estimation Perspective}

\author{Fan Liu, Tingting Zhang, {\it Member, IEEE}, Zenan Zhang, Bin Cao,  {\it Member, IEEE}, \\Yuan Shen, {\it Senior Member, IEEE}, and Qinyu Zhang, {\it Senior Member, IEEE}
\thanks{

Fan Liu, Tingting Zhang, Zenan Zhang, Bin Cao and Qinyu Zhang are with the School of Electronics and Information Engineering,
Harbin Institute of Technology, Shenzhen 518055, China. Tingting Zhang and Qinyu Zhang are also with Pengcheng Laboratory (PCL), Shenzhen 518055, China.
Fan Liu and Tingting Zhang are also with Guangdong Provincial Key Laboratory of Space-Aerial Networking and Intelligent Sensing, Shenzhen 518055, China.
 (e-mail: liufan0613@stu.hit.edu.cn; zhangtt@hit.edu.cn; h13zzn@126.com; caobin@hit.edu.cn; zqy@hit.edu.cn).

Yuan Shen is with the Department of Electronic Engineering, Tsinghua University, Beijing 100084, China, and also with the Beijing National Research Center for Information Science and Technology, Beijing 100084, China (e-mail: shenyuan$\_$ee@tsinghua.edu.cn).
}
}

\maketitle

\begin{abstract}

This paper investigates a bi-static integrated sensing and communication (ISAC) system for multi-target scenarios using impulse radio ultra-wideband (IR-UWB) signals, which offer fine temporal resolution, low power consumption, and strong resistance to multipath interference. Two typical modulation schemes, namely pulse position modulation (PPM) and binary phase shift keying (BPSK), are considered for communication over the delay and phase domains, respectively. Accordingly, we introduce a pilot-based decoupling approach that relies on known time-delays, as well as a differential decoupling strategy that uses a known starting symbol position ({no pilot required}). A key contribution of this work is the development of a unified analytical framework based on the Fisher Information Matrix (FIM), which characterizes the fundamental coupling between communication and sensing in both delay and Doppler domains. This coupling is examined through the singularity structure of the FIM, providing new theoretical insights into the joint performance limits of UWB-ISAC systems.
Finally, we assess the sensing and communication performance under various modulation schemes under the constraints of current UWB standards. This assessment utilizes the Cramer-Rao Lower Bound (CRLB) for sensing and the data transmission rate for communication, offering theoretical
insights into choosing suitable data signal processing methods in real-world applications.

\end{abstract}


\section{Introduction}

\IEEEPARstart{T}{he} increasing demands for spectrum, coupled with the limited bandwidth availability, have motivated the advancement of co-existing communication and radar system architectures\cite{LiuAn,FanLiu6G,zhangying}. Time- and frequency-division, as well as spatial beamforming, are common techniques for enabling communication and radar system coexistence. However, the integration of sensing and communication (ISAC) using identical waveforms for both functions provides more efficient spectrum use, reducing competition and optimizing energy and resource allocation without extra hardware. Additionally, this approach enhances adaptability to dynamic
environments through shared signal processing\cite{Waveform6G,JointWaveform}.

ISAC systems can be primarily categorized into continuous wave (CW)-based and pulse-based systems, respectively, depending on the type of signal waveforms used\cite{zhanghaojianotfs,tianxuanxuan,liufanMILCOM}. CW-based ISAC systems offer the benefit of high spectral efficiency for communication, but they are limited by low-resolution sensing capabilities. Additionally, self-interference poses a significant challenge for CW-based ISAC systems\cite{FMCW}.
Pulse-based systems have been extensively employed in target detection, localization, imaging and communication due to their benefits of low power consumption, high resolution, and strong resistance to narrowband interference\cite{PulseISAC,PulseISAC3}. Among popular pulse-based systems, it is worth noting that impulse radio ultrawide bandwidth (IR-UWB) signals have received considerable attention from both industry and academia in recent years\cite{UWBsurvey,UWBSY,UWBChannel,UWBshen}. Initially designed for radar sensing, IR-UWB signals were first explored for communication applications by Win and Scholtz\cite{how_it_works,MulAccessWin}.

Establishing a suitable theoretical framework for parameter estimation is essential for evaluating ISAC performance. The Cramer-Rao lower bound (CRLB) is pivotal in assessing the accuracy of local estimation tasks, as it defines the minimum variance that any unbiased estimator can
achieve\cite{Shen,HuangTY,zhang2019Tcom,CRLBISAC}. Numerous ISAC-related studies have used the CRLB as a benchmark for optimizing system parameters, including waveform parameter optimization, resource allocation, assessing the impact of signal processing methods on system performance,
etc\cite{CRLBXvJie,CRLBYuanWLiu,SymbioticISAC,zhang2016power}. However, there are still limited investigations on {\it how specific data modulation mechanisms affect the sensing performance}.

For pulse-based ISAC systems, a primary task is to measure the Doppler shift and time delay of moving targets, respectively. Consequently, it will lead to explicit {\it coupling} with data symbols using various modulation schemes\cite{UWBTR}.\footnote{In UWB cases, the target delay will couple with the pulse modulation position, while the Doppler shift will couple with the phase based modulation schemes, such as BPSK, etc.} Our previous work showed that data demodulation can be achieved by either adding pilot signals or using differential pulse positions
\cite{LiuFanTWC,ZhangZZICC,GCLiu}. This enables further realization of data demodulation and target sensing.

Notably, the latest proposals for the IEEE 802.15.4ab standard seriously consider the integration of sensing capabilities in addition to the communication and ranging functions. This marks a significant shift from the earlier UWB standards, such as IEEE 802.15.4a and 802.15.4z, which did not include sensing capabilities\cite{UWB4z,OverviewUWBStands}. Meanwhile, the UWB ecosystem has gradually been established, including applications in widely spread smartphones, smart homes, etc.

Since UWB technology primarily focus on short-range high-precision localization, sensing and low-rate communication applications, the modulation scheme selection involves critical trade-offs. Although pulse position modulation (PPM) and binary phase-shift keying (BPSK) exhibit lower spectral efficiency than quadrature phase-shift keying (QPSK) or orthogonal frequency-division multiplexing (OFDM), these modulation schemes offer the following key advantages for UWB systems: (1) lower implementation complexity, (2) enhanced resilience to multipath interference, (3) superior timing resolution for precision ranging~\cite{Saw}. These characteristics explain why PPM and BPSK remain the most widely adopted modulation schemes in UWB-based positioning systems and internet of things applications, as specified in major UWB standards including IEEE 802.15.4ab and 802.15.4z.

In this paper, we aim to investigate the impact of the coupling relationship between communication and sensing signals under the PPM and BPSK modulation schemes, in terms of {\it time-delay and phase}, on the sensing performance.
The main contributions of this work are summarized as follows:

\begin{dblparenenum}

\item We propose a bi-static sensing system with communication capabilities using integrated UWB signals. While employing PPM and BPSK modulation schemes for data transmission, this system also enables simultaneous dynamic target sensing via {\it delay} and {\it Doppler} estimation from time and phase measurements.

\item The coupling between communication and sensing in both delay and Doppler domains is jointly characterized within a unified framework based on the Fisher Information Matrix (FIM), where the interactions are revealed and analyzed through its singularity properties. This formulation offers a theoretical basis for the design and optimization of integrated UWB communication and sensing systems.

\item We propose an innovative decoupling strategy for the PPM-based UWB-ISAC system based on the position of the starting symbol (called {\it differential} decoupling, {\it no pilot required}). In comparison with the time-delay decoupling strategy (called {\it pilot-based} decoupling), analyses of sensing performance under various modulation \& demodulation schemes are presented.

\end{dblparenenum}

For clarity, the basic framework of this paper is shown in Fig.~\ref{BasicFramework}.
\begin{figure}[h]
\vspace{-8pt}
	\centering
	{\includegraphics[width = 0.7\columnwidth]{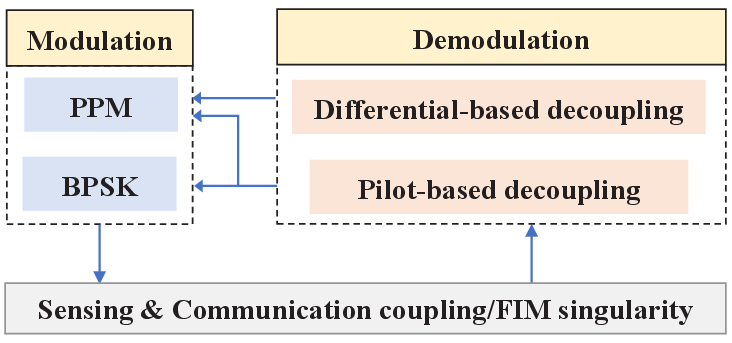}
\caption {Basic framework of this paper.}  \label{BasicFramework}}
  \vspace{-5pt}
\end{figure}

In this piece of work, we select pulse based UWB waveforms for the coupling analysis. Since data demodulation and target sensing are essentially parameter estimation problems, this proposed CRLB based framework can be extended to various ISAC waveforms, where the data symbols conveyed in the delay, phase and amplitude domains will also be coupled with the target parameters\cite{zhanghaojianotfs,OFDM}.

{\it Notations:} Symbol $*$ denotes the linear convolution; $[ \cdot ]_{m \times n}$ denotes an $m \times n$ matrix; $[\cdot]^{\rm T}$, $[\cdot]^{\rm H}$,
and $[\cdot ]^{\rm -1}$ denote the transpose, the conjugate transpose and the inverse of a matrix, respectively; tr $\{ \cdot \}$  denotes the trace of a square matrix.

\section{Preliminary}

\subsection{Channel Model for Bistatic ISAC Scenarios}  \label{Sec_channel}

In Fig.~\ref{FigISAC_scenario}, we depict a typical bi-static ISAC scenario, where the transmitter (Tx) sends {\it modulated} UWB pulses to the receiver (Rx). Within the designated area, multiple targets are present, capable of reflecting the transmitted signals. The Rx receives the resulting multipath
components and performs data demodulation and target sensing.
\begin{figure}[h]
\vspace{-5pt}
	\centering
	{\includegraphics[width = 0.7\columnwidth]{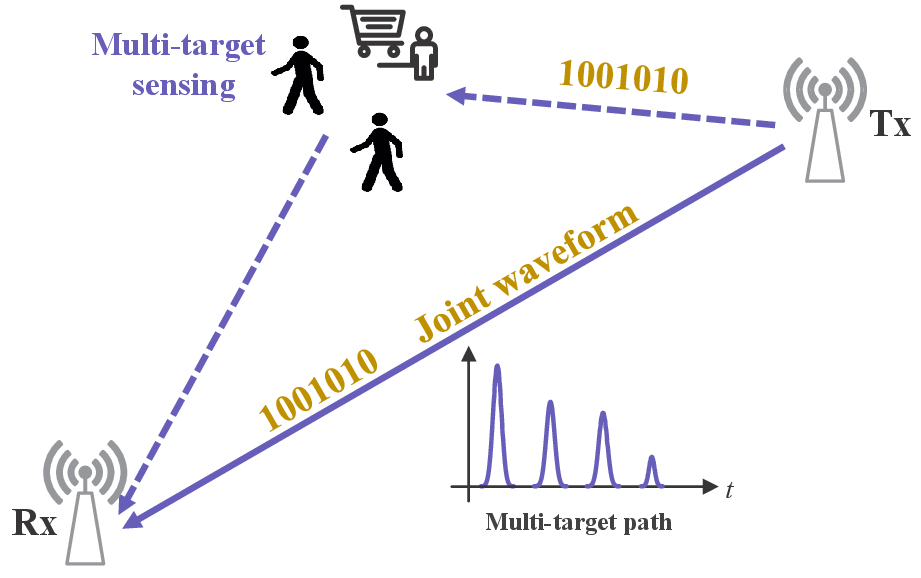}
\caption {Waveform integration sensing and communication system.} \label{FigISAC_scenario}}
\vspace{-5pt}
\end{figure}

As the most popular UWB channel models, the IEEE 802.15.4a model was proposed long time ago based on the S-V channel model, with extensive empirical data fitting\cite{MolisChannel}.
The corresponding channel impulse response can be expressed as follows:
\vspace{-4pt}
\begin{equation}  \label{SVchannel}
h_{\rm 4a}\left( t \right) = \sum\limits_{l = 1}^L {\sum\limits_{k = 1}^K {{\alpha _{k,l}}\delta \left( {t - {\tau_l} - {T_{k,l}}} \right)} } {e^{j{\phi _{k,l}}}},
\end{equation}
where ${\alpha_{k,l}}$ and ${\phi_{k,l}}$ represent the amplitude complex gain and the phase of the $k$-th multipath component in the $l$-th cluster, ${\tau_l}$ denotes the arrival time of the multipath in the $l$-th cluster, and ${T_{k,l}}$ represents the arrival time of the $k$-th multipath component within the $l$-th cluster relative to ${\tau_l}$.

Notably, the 4a channel model was primarily intended for {\it communications} and ranging, usually in {\it static} environments. It is usually characterized by its intense dispersion in the delay domain, without considering any Doppler effects. However, for the target delay and Doppler sensing, the channel impulse response of interest is expected to be {\it sparse} \cite{RDCFAR}. This sparsity arises from the fact that received signals exhibiting correlation - typically within a single cluster - are assumed to originate from the same target. Therefore, in the ISAC scenario discussed in this paper, {\it we are only interested in the arrival times of clusters $\tau_l$}\cite{molisch2011}.
Furthermore, since we address the sensing of {\it dynamic} targets in this paper, the standard 4a channel model will be extended to incorporate the estimation of Doppler shifts in the received signals.

Despite the dense multipath components of UWB propagations, the number of clusters is limited. In typical indoor environments described in the IEEE 802.15.4a standard, the number of clusters $L$ generally ranges from 3 to 5\cite{MolisChannel}. The cluster arrival times can be estimated using clustering algorithms, such as affinity propagation\cite{cluster1}. The extracted components retain the primary energy of the original signals, which facilitates communication data demodulation. Therefore, the proposed ISAC channel model can still be expressed as a series of Dirac functions\cite{ZhaoNaTcom}, each characterized by distinct amplitudes, delays, and phases, such as\footnote{This applies only when the absolute bandwidth is ultrawideband (greater than 500 MHz in the UWB system), but the relative bandwidth is small (less than approximately $20\%$).}
\vspace{-3pt}
\begin{equation}    \label{EqChannel}
\begin{aligned}
h(t) & = \sum_{l=1}^L\alpha^{\kappa}_{l}\delta(t-\tau^{\kappa}_{l}){e^{-{\rm{j}}2\pi {f_{\rm{c}}}{\tau^{\kappa}_l}}}     \\
     & \approx \sum_{l=1}^L {\alpha^{\kappa}_l} \delta(t - {\tau^{\kappa}_l}){e^{{\rm{j}}{2\pi {f_{{\text d l}}}({\kappa}T_{\text f})}}}
     {e^{-{\rm{j}}2\pi {{f_{\rm{c}}{\tau _{l0}}}}}},
\end{aligned}
\end{equation}
where $L$ represents the number of channel clusters (equivalently, the target number within the sensing area), $f_{\rm c}$ denotes the carrier frequency,
$t\in \left[ { {\kappa} {T_{\text f}},{(\kappa+1)}{T_{\text f}}} \right)$, $\kappa = 0, \cdots, N_{\text f}-1$, $N_{\text f}$ is the total number of signal transmission periods,{\footnote{We define the sensing period as the product of ${N_{\text f}}{T_{\text f}}$, representing the duration over which a single sensing process is completed.} $T_{\text f}$ is the pulse repetition interval (PRI), ${\alpha^{\kappa}_l}$ and ${\tau^{\kappa}_l}$ are the attenuation and signal propagation delay from the Tx to the Rx of the $l$-th path in the $\kappa$-th PRI, ${\tau^{\kappa}_l} = {\tau _{l0}} - {{{v_{\text{d} l}}} \over c}t$, ${\tau _{l0}}$ is the signal propagation delay from Tx to Rx of the $l$-th path in the first PRI, ${v_{\text{d}l}}$ is the radial velocity of the $l$-th target with respect to the Rx, $c$ is the speed of electromagnetic waves in the air, ${f_{\text {d}l}} = {f_{\rm{c}}}{{{v_{\text d l}}} \over c}$ is the Doppler shift arise from the target,
the approximation condition holds due to the Doppler shift being negligibly small within one PRI.

Fig.~\ref{Channel} presents the results of indoor channel measurements using commercial UWB modules, forming a bistatic scenario as Fig.~\ref{FigISAC_scenario}. There exists a direct path and two human targets in the sensing area, resulting in three clusters in the received waveform. The number of $L$ can be estimated using the clustering algorithm in \cite{cluster1}.\footnote{As this paper is primarily concerned with the number of $L$ in the derivation of the performance lower bound, the clustering algorithm will not be elaborated here.}
\begin{figure}
  \centering
  \includegraphics[width = 1\columnwidth]{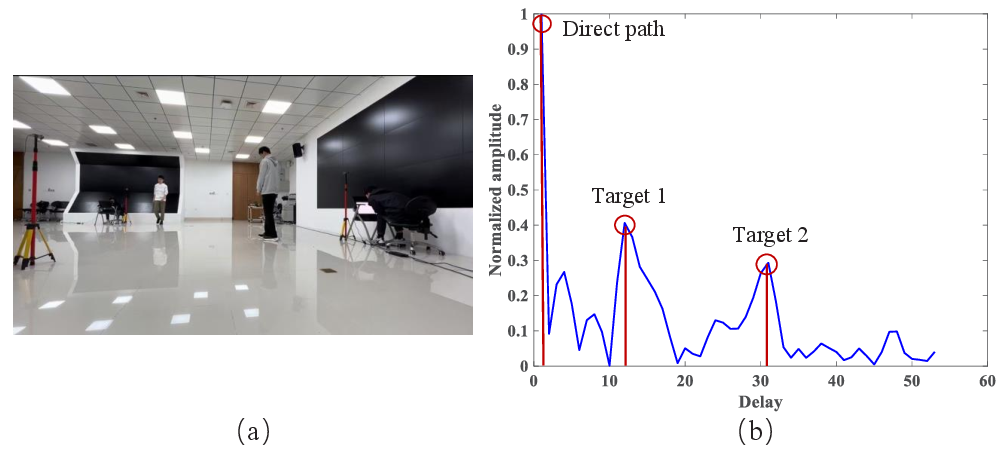}
  \vspace{-18pt}
  \caption{UWB environment measurement. The red lines in Fig.~(b) correspond to the channel impulse response that have extracted from three clusters, i.e., $L=3$.}  \label{Channel}
  \vspace{-10pt}
\end{figure}

\subsection{Signal Model}
\subsubsection{Sensing signal model}    \label{Sec sening}

Generally, the periodic UWB based sensing signals transmitted from the Tx during the $\kappa$-th pulse repetition interval (PRI) after up-conversion can be written as
\vspace{-2pt}
\begin{equation}   \label{Eq st}
s(t) = \sqrt {{E_{{\rm{tb}}}}} w(t-{\kappa} {T_{\text f}}){e^{{\rm{j}}2\pi {f_{\rm{c}}}t}},
\end{equation}
where ${{E_{{\rm{tb}}}}}$ is the total transmitted energy of one pulse, $w(t)$ indicates the energy normalized Gaussian pulse, i.e.,
\begin{equation}\label{GaussianPulse}
w(t) = \left( \frac{1}{\alpha \sqrt{\pi}} \right)^{1/2}\exp\left(-\frac{t^2}{2\alpha^2}\right),
\end{equation}
where $\alpha$ represents the temporal spreading factor of the Gaussian pulse.

The received UWB sensing signals in the $\kappa$-th PRI after down-conversion can be expressed as
\vspace{-2pt}
\begin{equation}      \label{EqRXSensing}
\begin{aligned}
r^{\kappa}_{\rm {s}}(t)& = s(t) *  h(t) + z(t)     \\
& \approx \sum_{l=1}^L {\tilde \alpha^{\kappa}_l} w(t - {\tau^{\kappa}_l} - {\kappa}{T_{\text f}}){e^{{\rm{j}}{\phi_l^{\kappa}}}} + z(t),
\end{aligned}
\end{equation}
where $\tilde \alpha^{\kappa}_{l} = {\alpha^{\kappa}_l} \sqrt{E_{\rm{tb}}}$,
$\phi_l^{\kappa} = 2\pi \big({f_{{\text d l}}}({\kappa}T_{\text f}) {-{f_{\rm{c}}}{\tau _{l0}}} \big)$,
and the approximation condition holds due to the Doppler shift being negligibly small within one PRI, the corresponding observation vector about phase a single sensing process are
$\boldsymbol \phi = \big[\boldsymbol \phi^{0}, ..., \boldsymbol \phi^{\kappa}, ..., \boldsymbol \phi^{N_{\text f}-1} \big]^{\rm T}$ and
$\boldsymbol{\phi}^{\kappa}=[\phi^{\kappa}_{1}, ..., \phi^{\kappa}_{L}]^\text{T}$,
the observation vector about time-delay can be expressed as
$\boldsymbol \tau = \big[\boldsymbol \tau^{0}, ..., \boldsymbol \tau^{\kappa}, ..., \boldsymbol \tau^{N_{\text f}-1}\big]^{\rm T}$  and  $\boldsymbol \tau^{\kappa} = \big[\tau_{1}^{\kappa}, ...,\tau_{L}^{\kappa} \big]^{\rm T}$,
the observation vector about the amplitude of the received signal can be expressed as
$\tilde {\boldsymbol{\alpha}}=[\tilde{\boldsymbol \alpha}^{0}, ..., \tilde {\boldsymbol \alpha}^{\kappa}, ...,\tilde{\boldsymbol \alpha}^{N_{\text f}-1}]^\text{T}$
and $\tilde {\boldsymbol{\alpha}}^{\kappa}=[\tilde  \alpha^{\kappa}_{1}, ..., \tilde \alpha^{\kappa}_{L}]^\text{T}$,
$z(t)$ represents the additive white Gaussian noise with variance $\sigma^2$.

\subsubsection{ISAC signal model with PPM modulation}

In the proposed ISAC system described in this paper, the periodic UWB signals based PPM scheme transmitted from the Tx in the $\kappa$-th PRI after up-conversion can be expressed as
\vspace{-2pt}
\begin{equation}   \label{Eq stPPM}
s_{\rm ppm}(t) = \sqrt {{E_{{\rm{tb}}}}} w(t-{{\xi _{{\rm{ppm}}}}{q_{\kappa}}}-{\kappa} {T_{\text f}}){e^{{\rm{j}}2\pi {f_{\rm{c}}}t}},
\end{equation}
where $q_\kappa$ is the ${\kappa}$-th deterministic unknown signal to be transmitted, ${\xi _{{\rm{ppm}}}}$ is the time-shift of PPM when $q_\kappa = 1$, the pulse position remains unchanged when $q_\kappa = 0$.

Then the received UWB ISAC signals based on PPM within the $\kappa$-th PRI after frequency down-conversion can be written as
\begin{equation}      \label{EqRXPPM}
\begin{aligned}
& r^{\kappa}_{\rm {ppm}}(t)  = s_{\rm ppm}(t) * h(t) + z(t)               \\
& \approx \!\! \sum_{l=1}^L \!  {\tilde \alpha^{\kappa}_l} w(t\! -\! {\underbrace {\tau _l^\kappa \! - \! \Delta \tau _{\rm q}^\kappa}_{{\rm{coupled}}}} \! -\!
\kappa T_{\text f}){e^{{\rm{j}}2\pi \big(\! {f_{{\text d l}}}({\kappa}T_{\text f}) {-{f_{\rm{c}}}{\tau _{l0}}} \! \big)}} \!\! + \!\! z(t)             \\
& = \!\! \sum_{l=1}^L {\tilde \alpha^{\kappa}_l} w(t-\tau^{\kappa}_{l,{\rm ppm}} - \kappa T_{\text f}){e^{{\rm{j}}\phi_l^{\kappa}}} + z(t),
\end{aligned}
\end{equation}
where ${\Delta \tau^{\kappa}_{\rm q}} = {\xi_{{\rm{ppm}}}} {q_{\kappa}}$ represents the time delay caused by data modulation,
and $\tau^{\kappa}_{l, {\rm ppm}} = \tau _l^\kappa  + \Delta \tau _{\rm q}^\kappa$ is the total time delay observed.\footnote{We set $\tau^{\kappa}_{l, {\rm ppm}}$ as a constant but unknown parameter to be estimated during one deterministic transmission, which can be categorized as a {\it non-Bayesian}
estimation problem.}
The detailed derivation of \eqref{EqRXPPM} is consistent with equation \eqref{EqRXSensing} and will not be elaborated here.
It can be observed that {\bf the signal transmission delay $\tau _l^\kappa$ and the modulation data $\Delta \tau _{\rm q}^\kappa$ are coupled in the time domain.}
The corresponding observation vector can be expressed as
$\boldsymbol \tau_{\rm ppm} = \big[\boldsymbol \tau_{\rm ppm}^{0}, ..., \boldsymbol \tau_{\rm ppm}^{\kappa}, ..., \boldsymbol \tau_{\rm ppm}^{N_{\text f}-1}\big]^{\rm T}$  and  $\boldsymbol \tau_{\rm ppm}^{\kappa} = \big[\tau_{1,\rm ppm}^{\kappa}, ...,\tau_{L,\rm ppm}^{\kappa} \big]^{\rm T}$.

\subsubsection{ISAC signal model with BPSK modulation}

Similarly to \eqref{Eq stPPM}, the received UWB ISAC signals with BPSK modulation scheme after frequency down-conversion can be written as
\begin{equation}   \label{Eq stBPSK}
s_{\rm bpsk}(t) = \sqrt {{E_{{\rm{tb}}}}} w(t-{\kappa} {T_{\text f}}){e^{{\rm{j}}(2\pi {f_{\rm{c}}}t-\xi_{\rm bpsk} q_{\kappa})}},
\end{equation}
where ${\xi _{{\rm{bpsk}}}}$ is the phase-shift introduced by BPSK modulation when $q_\kappa = 1$, the transmitted signal phase remains unchanged when $q_\kappa = 0$.

The received UWB ISAC signals with BPSK modulation scheme within the $\kappa$-th PRI after frequency down-conversion can be expressed as
\begin{equation}         \label{EqRXBPSK}
\begin{aligned}
r^{\kappa}_{\rm bpsk}(t) & = s_{\rm bpsk}(t) * h(t) + z(t)          \\
& = \sum_{l=1}^L  {\tilde \alpha^{\kappa}_l}w(t -{\tau^{\kappa}_l}-\kappa{T_{\text f}}){e^{-{\rm{j}} 2 \pi f_{\rm c} {\tau _{l0}}}}
{e^{{\rm{j}}{\phi^{\kappa}_{l,{\rm bp}}}}} + z(t),
\end{aligned}
\end{equation}
where $\phi^{\kappa}_{l,{\rm bp}} = {2\pi \kappa {T_{\text f}}\Big( {\underbrace {{f_{{\rm{d}}l}}+ \varphi _{{\rm{bpsk}}}^\kappa }_{{\rm{coupled}}}} \Big)}$,
$\varphi _{{\rm{bpsk}}}^\kappa  = {{{\xi _{{\rm{bpsk}}}}{q_\kappa }} \over {2\pi \kappa {T_{\text f}}}}$  denotes the phase component related to the data, which lacks tangible physical significance. The observed phase of the $l$-th path can be written as
$\phi^{\kappa}_{l,{\rm bpsk}} = {\phi^{\kappa}_{l,{\rm bp}} -2 \pi f_{\rm c} {\tau _{l0}}}$,
and the corresponding observation vector is written as
$\boldsymbol \phi_{\rm bpsk} = \big[\boldsymbol \phi_{\rm bpsk}^{0}, ..., \boldsymbol \phi_{\rm bpsk}^{\kappa}, ..., \boldsymbol \phi_{\rm bpsk}^{N_{\text f}-1} \big]^{\rm T}$  and  $\boldsymbol \phi_{\rm bpsk}^{\kappa} = \big[\phi_{1,\rm bpsk}^{\kappa}, ...,\phi_{L,\rm bpsk}^{\kappa} \big]^{\rm T}$.
We can see that {\bf the Doppler shift ${f_{{\rm{d}}l}}$ and the phase component $\varphi _{{\rm{bpsk}}}^\kappa$ related to the modulation data are coupled in the phase domain.}

\subsection{Signal Re-representation}

In this subsection, we provide discretized representations of the received signals in Sec.~\ref{Sec sening}. This prepares the groundwork for deriving the theoretical model in this paper.

In the parametric representation of \eqref{EqRXSensing}, the signals are typically expressed as time-delayed versions of $w(t)$ with the corresponding unknown time-delays $\tau_l$, i.e.,
\vspace{-3pt}
\begin{equation}
\begin{aligned}
  {r^\kappa_{\rm s} }(t) &  \approx \! \sum\limits_{l = 1}^L \! {{{\tilde \alpha^{\kappa} }_l}} w({\tau^{\kappa}_l},t \! - \! \kappa {T_{\text f}}){e^{{\rm{j}} \phi^\kappa_l}}\! + \! z(t),         t \! - \! \kappa {T_{\text f}} \in \left[ {0, {T_{\text f}}} \right),
\end{aligned}
\end{equation}
where $w({\tau^{\kappa}_l},t) = w(t - {\tau^{\kappa}_l})$.

To facilitate a more intuitive understanding of the signal processing procedure, the signal sampling process can be described as a matrix form. The vector related to time-delay is defined as
\begin{equation}  \label{EqWtN}
\begin{aligned}
{\boldsymbol w}\left( {\tau _l^\kappa } \right){\rm{ = }} & \Big[ {w\left( {\tau _l^\kappa ,0} \right), w\left( {\tau _l^\kappa ,{T_{\rm{s}}}} \right),
\cdots, w\left( {\tau _l^\kappa ,\left( {{N_{\rm{s}}} - 1} \right){T_{\rm{s}}}} \right)} \Big]^{\text T}   \\
& \in {{\Bbb {R}}^{{N_{\rm{s}}} \times 1}},
\end{aligned}
\end{equation}
where $N_{\rm s}$ is the sampling points in one PRI.

Additionally, the vector associated with the phase across multiple PRIs is defined as
\begin{equation}   \label{EqDopplerN}
\boldsymbol d\left( {{\phi_l}} \right) = {\left[ {{e^{{\rm{j}}{\phi^0_l}}}, ..., {e^{{\rm{j}} \phi^{{N_{\text f}} - 1}_l}}} \right]^{\rm{T}}} \in {{\Bbb {C}}^{{N_{\text{f}}} \times 1}}.
\end{equation}
Subsequently, the vector of received signals without noise (mean of the received signals over multiple PRIs) can be expressed as \eqref{EqRxMean},
\begin{figure*} [b]
\vspace{-2pt}
\hrulefill
\vspace{-2pt}
\begin{equation}    \label{EqRxMean}
\boldsymbol \mu_{\rm s}  = {\left[ {\sum\limits_{l = 1}^L {{\tilde \alpha^{0} _l} {\boldsymbol d} {{\left( {\phi_l} \right)}_0}
{\boldsymbol w} \left( {\tau _l^0} \right)} ,  \cdots , \sum\limits_{l = 1}^L {{\tilde \alpha^{N_{\text f}-1} _l}{\boldsymbol d}{(\phi_l)_{{N_{\text f}} - 1}} {\boldsymbol w} \left( {\tau _l^{{N_{\text f}} - 1}} \right)} } \! \right]^{\text{T}}}   \\
\in {{\Bbb{C}}^{{N_{\text f}}{N_{\text{s}}} \times 1}},
\end{equation}
\vspace{-15pt}
\end{figure*}
where ${\boldsymbol d}{{\left({\phi_l} \right)}_0 = {e^{{\rm{j}}{\phi^0_l}}}}$.

Similarly, the received signals in both PPM and BPSK systems can be derived using the same approach. We express them as \eqref{EqRxMeanPPM} and \eqref{EqRxMeanBPSK},
\begin{figure*} [b]
\begin{equation}    \label{EqRxMeanPPM}
\boldsymbol \mu_{\rm ppm}  = {\left[ {\sum\limits_{l = 1}^L {{\tilde \alpha^{\kappa} _l} {\boldsymbol d} {{\left( {\phi_l} \right)}_0}
{\boldsymbol w} \left({\tau_{l,{\rm ppm}}^0} \right)} ,  \cdots , \sum\limits_{l = 1}^L {{\tilde \alpha^{N_{\text f}-1} _l}{\boldsymbol d}{(\phi_l)_{{N_{\text f}} - 1}} {\boldsymbol w} \left( {\tau _{l,{\rm ppm}}^{{N_{\text f}} - 1}} \right)} } \right]^{\text{T}}} \in {{\Bbb{C}}^{{N_{\text f}}{N_{\text{s}}} \times 1}},
\end{equation}
\begin{equation}   \label{EqRxMeanBPSK}
\boldsymbol \mu_{\rm bpsk}  = {\left[
{\sum\limits_{l = 1}^L {{\tilde \alpha^{\kappa} _l} {\boldsymbol d} {{\left( {\phi_{l,{\rm bpsk}}} \right)}_0} {\boldsymbol w} \left( {\tau _l^0} \right)} , \cdots,
\sum\limits_{l = 1}^L {{\tilde \alpha^{N_{\text f}-1} _l}{\boldsymbol d}{(\phi_{l,{\rm bpsk}})_{{N_{\text f}} - 1}} {\boldsymbol w} \left( {\tau _l^{{N_{\text f}} - 1}} \right)} }\right]^{\text{T}}} \in {{\Bbb{C}}^{{N_{\text f}}{N_{\text{s}}} \times 1}},
\end{equation}
\end{figure*}
where ${\boldsymbol w} \big({\tau_{l,{\rm ppm}}^0} \big)$ is the vector associated with the time-delay in the PPM ISAC system, similar to $\eqref{EqWtN}$,
and $\boldsymbol d\big( {{\phi_{l,{\rm bpsk}}}} \big)_0 = {e^{{\rm{j}}{\phi^0_{l,{\rm bpsk}}}}}$.

Before proceeding, we declare the received SNR for the $l$-th target as \eqref{EqSNR}, which will be utilized in deriving the CRLB.
\vspace{-3pt}
\begin{equation}       \label{EqSNR}
{SNR}^l = {{{{\left| {\tilde \alpha _l} \right|}^2} {\int_0^{T_{\text f}} \Big(w\big(t-{\tau _l^0}\big) \Big)^2} dt} / T_{\text f} \over {\sigma^2}}  \\
\end{equation}

Similarly, the effective bandwidth \cite{Shen} of the transmitted pulse is expressed as
\begin{equation}       \label{EqBandwidthFre}
B =\Bigg({{\int_{-\infty}^{\infty} {f^2}|{S{(f)}}|^2 } df  \over {\int_{-\infty}^{\infty} { |{S{(f)}}|^2 }} df}\Bigg)^{1/2},    \\
\end{equation}
where $S{(f)}$ represents the frequency domain representation of the transmitted signal $s(t)$.  Correspondingly, the expression for the effective bandwidth from the time-domain signal is given by
\vspace{-3pt}
\begin{equation}       \label{EqBandwidth}
B =\Bigg({{\int_0^{T_{\text f}} \Big({w^{(1)}\big(t-{\tau _l^0} \big) \Big)^2}} dt  \over 4\pi^2 {\int_0^{T_{\text f}} \Big({w \big(t-{\tau _l^0} \big) \Big)^2}} dt}\Bigg) ^{1/2},    \\
\end{equation}
where $w^{(1)}\big(t-{\tau _l^0} \big)$ is the abbreviation of $dw\big(t-{\tau _l^0} \big) \over dt$.

\subsection{Estimated Parameters of Different Modulation Schemes}     \label{SecParameters}

The key parameters we focus on while utilizing PPM in the ISAC system can be written as
\begin{equation}      \label{Eq theta ppm}
\boldsymbol\theta_{\rm ppm}= \Big [\tau_{\text 1}, \Delta \boldsymbol{\tau}^{\rm T}, \Delta \tau_{{\rm q}},
f_{{\rm d} 1},\Delta \boldsymbol{f}_{\rm d}^{\rm T},  \tilde {\boldsymbol\alpha}^{\rm T} \Big ]^{\rm T},
\end{equation}
where $\Delta \boldsymbol{\tau}= \big[\Delta \tau_{\rm 2}, ..., \Delta \tau_{L} \big]^\text{T}$, the term $\Delta \tau_{l} = \tau_{l} - \tau_{1}$ represents the time delay difference between the $l$-th path and the first path. Correspondingly,
$\Delta\boldsymbol{f}_{\rm d} = \big[\Delta f_{\rm d2},..., \Delta f_{{\rm d}{L}} \big]^\text{T}$, the term $\Delta f_{{\rm d}{l}} = f_{{\rm d}{l}} - f_{{\rm d}{1}}$ represents the Doppler shift difference between the $l$-th path and the first path.\footnote{Assume that the time-delay remains within
one time resolution $T_{\rm s}$ during the sensing period, meaning that $\tau_1^{0}= \ldots = \tau^{\smash{\raisebox{-0ex}{\tiny$N_f$}}}_1$, $\alpha_1^{0}= \ldots = \alpha^{\smash{\raisebox{-0ex}{\tiny$N_f$}}}_l$.}

Similarly, the parameters we are interested in while utilizing BPSK modulation in the ISAC system can be written as follows:
\begin{align}	          \label{Eqthetabpsk}
\boldsymbol\theta_{{\rm bpsk}} = \Big [\tau_{\text 1},\Delta \boldsymbol{\tau}^{\rm T},f_{{\rm d}1},\Delta \boldsymbol{f}_{\rm d}^\text{T},  \varphi _{\rm{bpsk}}, \tilde{\boldsymbol\alpha}^{\rm T} \Big ]^{\rm T}.
\end{align}

In the sensing-only scenario, the parameter vectors are given in equations \eqref{Eq theta ppm} and \eqref{Eqthetabpsk} can be simplified to
\begin{equation}          \label{Eq theta Sensing}
\boldsymbol\theta_{\rm s} = \Big[ \tau_{\rm 1},\Delta \boldsymbol{\tau}^{\rm T},f_{{\rm d}1},\Delta \boldsymbol{f}_{\rm d}^\text{T}, \tilde{\boldsymbol\alpha}^{\rm T} \Big]^{\rm T},
\end{equation}
which is obtained by eliminating $\Delta \tau_{\rm{q}}$ and $\varphi _{{\rm{bpsk}}}$ from  the parameter vectors $\boldsymbol\theta_{\rm{ppm}}$ and $\boldsymbol\theta_{\rm{bpsk}}$.

\vspace{-10pt}
\subsection{CRLB on Parameters Estimation}    \label{SecCRLBModel}

In parameter estimation theory, the observation parameters vector is a crucial component, reflecting the different parameters that the system needs to estimate for effective performance. The observation parameters vector related to \eqref{EqRXSensing} in the sensing only case can be expressed as
\begin{equation}       \label{EqSensingparemeter}
{\boldsymbol {\eta}}={\Big[{\boldsymbol \tau ^{\rm{T}}},{\boldsymbol \phi}^{\rm{T}},{\tilde {\boldsymbol \alpha} ^{\rm{T}}} \Big]^{\rm{T}}},
\end{equation}
where the detailed expression of the parameters is provided after \eqref{EqRXSensing}.

Then the FIM of \eqref{EqSensingparemeter} can be expressed as
\begin{equation}   \label{EqIeta}
\mathbf{I}_{{\boldsymbol\eta}}=
\left[\begin{array} {lll}
\boldsymbol\Lambda_{\boldsymbol \tau, \boldsymbol \tau}
& \boldsymbol\Lambda_{\boldsymbol \tau, {\boldsymbol \phi}}
& \boldsymbol\Lambda_{{\boldsymbol \tau},{ \tilde {\boldsymbol \alpha}}} \\
\boldsymbol\Lambda_{{\boldsymbol \phi}, \boldsymbol \tau}
& \boldsymbol\Lambda_{{\boldsymbol \phi},  {\boldsymbol \phi}}
& \boldsymbol\Lambda_{{\boldsymbol \phi}, \tilde{\boldsymbol \alpha}} \\
\boldsymbol\Lambda_{\tilde {\boldsymbol \alpha}, {\boldsymbol \tau}}
& \boldsymbol\Lambda_{\tilde{\boldsymbol \alpha}, {\boldsymbol \phi}}
& \boldsymbol\Lambda_{\tilde{\boldsymbol \alpha}, \tilde{\boldsymbol \alpha}}
\end{array}\right],
\end{equation}
where $\boldsymbol\Lambda_{\boldsymbol \tau, \boldsymbol \tau}$ represents the sub-matrix about time-delay. Each element in $\mathbf{I}_{{\boldsymbol\eta}}$ can be derived from \eqref{Eqelement eta}, given by~\cite{Fundament},
\begin{equation}   \label{Eqelement eta}
 \mathbf{I}_{{\boldsymbol\eta}}(i_1,j_1)=\frac{2}{\sigma^{2}} \operatorname{Re}
  \left\{\frac{\partial \boldsymbol{\mu}_{\rm s}^{\rm{H}}}{\partial \boldsymbol \eta_{i_1}} \frac{\partial \boldsymbol{\mu}_{\rm s}}{\partial \boldsymbol\eta_{j_1}}\right\},
\end{equation}
where $\boldsymbol \eta_{i_1}$ denotes the ${\it i_{\rm 1}}$-th element of $\boldsymbol \eta$.

Then the FIM of estimated vectors $\boldsymbol \theta_{\rm s}$ is given by
\begin{equation}   \label{EqEFIMSense}
\mathbf{I}_{{\boldsymbol \theta}_{\rm s}}= \mathbf {J}_{\rm s}^{\text T} {\mathbf I}_{\boldsymbol\eta} \mathbf {J}_{\rm s},
\end{equation}
where $\mathbf {J}_{\rm s}$ is the Jacobian matrix that maps ${\boldsymbol\eta}$ to ${\boldsymbol \theta_{\rm s}}$, describing how the components of ${\boldsymbol \theta_{\rm s}}$ vary with respect to changes in ${\boldsymbol\eta}$.
The elements in $\mathbf{I}_{{\boldsymbol \theta}_{\rm s}}$ can be divided into several blocks according to the estimated parameter vector
${\boldsymbol \theta_{\rm s}}$, as follows:
{\setlength{\arraycolsep}{3pt}
\begin{equation}    \label{Eqelement etaRe}
{{\bf{I}}_{{\boldsymbol \theta _{\rm{s}}}}} = \left[ \begin{array}{c:cccc}
   {{{\bf{\Lambda }}_{{\tau _1},{\tau _1}}}} \hfill
   & {{{\bf{\Lambda }}_{{\tau _1},\Delta \tau }}} \hfill
   & {{{\bf{\Lambda }}_{{\tau _1},{f_{{\rm{d1}}}}}}} \hfill
   & {{{\bf{\Lambda }}_{{\tau _1},\Delta {f_{\rm{d}}}}}} \hfill
   & {{{\bf{\Lambda }}_{{\tau _1},\tilde \alpha }}} \hfill  \\
   \hdashline
   {{{\bf{\Lambda }}_{\Delta \tau ,{\tau _1}}}} \hfill
   & {{{\bf{\Lambda }}_{\Delta \tau ,\Delta \tau }}} \hfill
   & {{{\bf{\Lambda }}_{\Delta \tau ,{f_{{\rm{d1}}}}}}} \hfill
    & {{{\bf{\Lambda }}_{\Delta \tau ,\Delta {f_{\rm{d}}}}}} \hfill
    & {{{\bf{\Lambda }}_{\Delta \tau ,\tilde \alpha }}} \hfill  \\
   {{{\bf{\Lambda }}_{{f_{{\rm{d1}}}},{\tau _1}}}} \hfill
   & {{{\bf{\Lambda }}_{{f_{{\rm{d1}}}},\Delta \tau }}} \hfill
   & {{{\bf{\Lambda }}_{{f_{{\rm{d1}}}},{f_{{\rm{d1}}}}}}} \hfill
   & {{{\bf{\Lambda }}_{{f_{{\rm{d1}}}},\Delta {f_{\rm{d}}}}}} \hfill
   & {{{\bf{\Lambda }}_{{f_{{\rm{d1}}}},\tilde \alpha }}} \hfill  \\
   {{{\bf{\Lambda }}_{\Delta {f_{\rm{d}}},{\tau _1}}}} \hfill
   & {{{\bf{\Lambda }}_{\Delta {f_{\rm{d}}},\Delta \tau }}} \hfill
   & {{{\bf{\Lambda }}_{\Delta {f_{\rm{d}}},{f_{{\rm{d1}}}}}}} \hfill
   & {{{\bf{\Lambda }}_{\Delta {f_{\rm{d}}},\Delta {f_{\rm{d}}}}}} \hfill
   & {{{\bf{\Lambda }}_{\Delta {f_{\rm{d}}},\tilde \alpha }}} \hfill  \\
   {{{\bf{\Lambda }}_{\tilde \alpha ,{\tau _1}}}} \hfill
   & {{{\bf{\Lambda }}_{\tilde \alpha ,\Delta \tau }}} \hfill
   & {{{\bf{\Lambda }}_{\tilde \alpha ,{f_{{\rm{d1}}}}}}} \hfill
   & {{{\bf{\Lambda }}_{\tilde \alpha ,\Delta {f_{\rm{d}}}}}} \hfill
   & {{{\bf{\Lambda }}_{\tilde \alpha ,\tilde \alpha }}} \hfill  \\
  \end{array}  \right].
\end{equation}
}
\begin{Definition}   \label{DefEFIM}
The equivalent FIM (EFIM) for $\tau_1$ is given by
\begin{equation}\label{EqEFIMTau1}
{{\bf{I}}_{\tau_1}} = {\bf A} - {\bf B}^{\rm T}{{\bf C}^{ - 1}}{\bf B},
\end{equation}
where ${\bf A} = {{{\bf{\Lambda }}_{{\tau _1},{\tau _1}}}}$, ${\bf B} = [{{{\bf{\Lambda }}_{{\tau _1},\Delta \tau }}},
{{{\bf{\Lambda }}_{{\tau _1},{f_{{\rm{d1}}}}}}}, {{{\bf{\Lambda }}_{{\tau _1},\Delta {f_{\rm{d}}}}}},
{{{\bf{\Lambda }}_{{\tau _1},\tilde \alpha }}} ]^{\text T}$, and
$${\bf C} =
\left[ \begin{array}{cccc}
 {{{\bf{\Lambda }}_{\Delta \tau ,\Delta \tau }}} \hfill
    & {{{\bf{\Lambda }}_{\Delta \tau ,{f_{{\rm{d1}}}}}}} \hfill
    & {{{\bf{\Lambda }}_{\Delta \tau ,\Delta {f_{\rm{d}}}}}} \hfill
    & {{{\bf{\Lambda }}_{\Delta \tau ,\tilde \alpha }}} \hfill  \\
 {{{\bf{\Lambda }}_{{f_{{\rm{d1}}}},\Delta \tau }}} \hfill
   & {{{\bf{\Lambda }}_{{f_{{\rm{d1}}}},{f_{{\rm{d1}}}}}}} \hfill
   & {{{\bf{\Lambda }}_{{f_{{\rm{d1}}}},\Delta {f_{\rm{d}}}}}} \hfill
   & {{{\bf{\Lambda }}_{{f_{{\rm{d1}}}},\tilde \alpha }}} \hfill  \\
 {{{\bf{\Lambda }}_{\Delta {f_{\rm{d}}},\Delta \tau }}} \hfill
   & {{{\bf{\Lambda }}_{\Delta {f_{\rm{d}}},{f_{{\rm{d1}}}}}}} \hfill
   & {{{\bf{\Lambda }}_{\Delta {f_{\rm{d}}},\Delta {f_{\rm{d}}}}}} \hfill
   & {{{\bf{\Lambda }}_{\Delta {f_{\rm{d}}},\tilde \alpha }}} \hfill  \\
 {{{\bf{\Lambda }}_{\tilde \alpha ,\Delta \tau }}} \hfill
   & {{{\bf{\Lambda }}_{\tilde \alpha ,{f_{{\rm{d1}}}}}}} \hfill
   & {{{\bf{\Lambda }}_{\tilde \alpha ,\Delta {f_{\rm{d}}}}}} \hfill
   & {{{\bf{\Lambda }}_{\tilde \alpha ,\tilde \alpha }}} \hfill  \\
  \end{array}  \right].$$
\end{Definition}
\begin{Definition}
The CRLB of $\tau_1$ is defined to be
\begin{equation}    \label{EqCRLBTau}
{\cal C}(\tau_1) = {\rm{tr}}\Big[ {{\bf{I}}_{\tau_1}^{ - 1}} \Big].
\end{equation}

Accordingly, the CRLB of the signal transmission range from Tx to Rx can be written as
\begin{equation}   \label{EqCRLBRange}
{\cal C}(d_1) = c^2 {\cal C}(\tau_1),
\end{equation}
where $c$ is the speed of electromagnetic waves in the air, the number of clusters $L$ in ${{\bf{I}}_{\tau_1}}$ can be obtained through pre-measurement and processing of the channel as mentioned in Sec.~\ref{Sec_channel}.
\end{Definition}

The CRLB solutions for the other estimated parameters indicated in \eqref{Eq theta Sensing} follow a comparable methodology, differing only in the reordering of rows and columns within the information matrix  ${\bf I_{\boldsymbol \theta_{\rm s}}}$ as illustrated in \eqref{Eqelement etaRe}.

\section{Estimation Accuracy Based ISAC Problem Formulation}

This section derives the deterministic CRLB for estimating time-delay and Doppler shift within the UWB ISAC system. We pay particular attention to the interaction between the unknown modulated data and the sensing parameters, highlighting the coupling of information in the ISAC system.
Furthermore, we propose a decoupling method to obtain CRLB expressions for time-delay and Doppler under different modulation schemes, ensuring that we accurately quantify the limits of parameter estimation for various signal configurations.
  \vspace{-12pt}
\subsection{FIM in Sensing-only Case}

In this subsection, we consider the sensing-only scenario, where the system performs target sensing without any data transmission as Sec.~\ref{SecCRLBModel}.

\begin{Corollary}     \label{CorolEFIMSense}
The FIM of sensing parameter vector $\boldsymbol \theta_{\rm s}$ can be expressed as
\begin{equation}   \label{EqEFIMSense}
\begin{aligned}
&  {{\bf I}_{{\boldsymbol \theta}_{\rm s}}} = {\bf J}_{{\rm{s}}}^{\rm{T}}{{\bf{I}}_{\boldsymbol \eta_{\rm s}} }{{\bf J}_{{\rm{s}}}}      \\
& = {\left[ {\begin{array}{*{20}{c}}
   {{{\bf{H}}^{\rm{T}}}{{\bf{\Lambda }}_{\boldsymbol \tau_{\rm s}, \boldsymbol \tau_{\rm s} }}{\bf{H}}}
& {\bf 0}
& {{{\bf{H}}^{\rm{T}}}{{\bf{\Lambda }}_{\boldsymbol \tau_{\rm s}, {\tilde {\boldsymbol \alpha}}_{\rm s}}}}    \\
  {\bf 0}
& \mathfrak{b}{{{\bf{H}}^{\rm{T}}}{{\bf{\Lambda }}_{\boldsymbol \phi^{0},  \boldsymbol \phi^{0}}}{\bf{H}}}
& {\bf 0}                     \\
  {{\bf{\Lambda }}_{{\tilde {\boldsymbol \alpha}}_{\rm s}, \boldsymbol \tau_{\rm s}}} {\bf{H}}
& {\bf 0}
& {{\bf{\Lambda }}_{{\tilde {\boldsymbol \alpha}}_{\rm s}, {\tilde {\boldsymbol \alpha}}_{\rm s}}}
\end{array}} \right]},
\end{aligned}
\end{equation}
where ${\boldsymbol \eta_{\rm s}}={\boldsymbol \eta}$, $\mathbf {J}_{\rm s}$ is the Jacobian matrix from ${\boldsymbol\eta_{\rm s}}$ to ${\boldsymbol \theta_{\rm s}}$, $\mathfrak{b}{=} \frac{{2{\pi^2}{T_{\text f}^2}{N_{\text f}}\left( {{N_{\text f}} - 1} \right)\left( {2{N_{\text f}} - 1} \right)}}{3}$, and
\begin{equation}
\begin{aligned}
{\bf{H}} = {\left[ {\begin{array}{*{20}{c}}
 1 & 0 & \ldots & 0 \\
 1 & 1 & \ldots & 0 \\
 \vdots & \vdots & \ddots & \vdots \\
 1 & 0 & \ldots & 1
\end{array}} \right]_{L \times L}}.
\end{aligned}
\end{equation}

We observed a zero vector in the matrix when $N_f=1$, indicates that the matrix is singular. This aligns with the principle that Doppler estimation requires at least two periods of the signal.

\end{Corollary}

\begin{proof}
See Appendix \ref{AppendixProEFIMSense}.
\end{proof}

\vspace{-15pt}

\subsection{FIM in ISAC Systems with PPM Scheme}

The observation parameters vector in the PPM ISAC system is shown as
\begin{equation}   \label{EqthetaPPM}
{\boldsymbol {\eta}_{\rm ppm}}={\Big [{\boldsymbol \tau_{\rm ppm}^{\rm{T}}},{\boldsymbol \phi}^{\rm{T}},{\tilde {\boldsymbol \alpha}^{\rm{T}}} \Big ]^{\rm{T}}},
\end{equation}
where the detailed expression of $\boldsymbol{\tau}_{\rm ppm} $ is given in \eqref{EqRXPPM}.

\begin{proposition}     \label{PropoEFIMPPM}

The FIM of estimated parameter vector ${\boldsymbol \theta}_{\rm ppm}$ in \eqref{Eq theta ppm} can be expressed as
\begin{equation}        \label{EqEFIMPPM}
\begin{aligned}
& {{\bf I}_{{\boldsymbol \theta}_{\rm ppm}}} = {\bf J}_{{\rm{ppm}}}^{\rm{T}}{{\bf{I}}_{\boldsymbol \eta_{\rm ppm}}}{{\bf J}_{{\rm{ppm}}}}      \\
& = \! \!{\left[ \! {\begin{array}{*{20}{c}}
  {{{\bf{H}}^{\rm{T}}}{{\bf{\Lambda }}_{\boldsymbol \tau, \boldsymbol \tau}}{\bf{H}}}
&  {{{\bf{H}}^{\rm{T}}}{{\bf{\Lambda }}_{\boldsymbol \tau, \boldsymbol \tau}}{\bf{E}}}
&  {\bf 0}
& {{{\bf{H}}^{\rm{T}}}{{\bf{\Lambda }}_{\boldsymbol \tau, {\tilde {\boldsymbol \alpha}}}}}    \\
  {{{\bf{E}}^{\rm{T}}}{{{\bf{\Lambda }}_{\boldsymbol \tau,{\boldsymbol \tau}}}}{\bf{H}}}
& {{{\bf{E}}^{\rm{T}}}{{\bf{\Lambda }}_{\boldsymbol \tau,{\boldsymbol \tau}}}{\bf{E}}}
& {\bf 0}
& {{{\bf{E}}^{\rm{T}}}{{\bf{\Lambda }}_{\boldsymbol \tau, {\tilde {\boldsymbol \alpha}}}}}    \\
  {\bf 0}
& {\bf 0}
& \mathfrak{b}{{{\bf{H}}^{\rm{T}}}{{\bf{\Lambda }}_{\boldsymbol \phi^{0},  \boldsymbol \phi^{0}}}{\bf{H}}}
& {\bf 0}                         \\
  {{\bf{\Lambda }}_{{\tilde {\boldsymbol \alpha}}, \boldsymbol \tau}} {\bf{H}}
& {{\bf{\Lambda }}_{{\tilde {\boldsymbol \alpha}}, \boldsymbol \tau}} {\bf{E}}
& {\bf 0}
& {{\bf{\Lambda }}_{{\tilde {\boldsymbol \alpha}}, {\tilde {\boldsymbol \alpha}}}}
\end{array}} \! \right]\!\!,}
\end{aligned}
\end{equation}
where ${{\bf{I}}_{\boldsymbol \eta_{\rm ppm}}}$ is the FIM of $\boldsymbol \eta_{\rm ppm}$, and the numerical values of the elements in
${{\bf I}_{{\boldsymbol \eta}_{\rm ppm}}}$ are consistent with those in ${{\bf I}_{\boldsymbol \eta_{\rm s}}}$, ${\bf J}_{{\rm{ppm}}}$ is the Jacobian matrix which is mapping from $\boldsymbol \eta_{\rm ppm}$ to $\boldsymbol \theta_{\rm ppm}$,
\begin{equation}
{\bf{E}} = {[ {\begin{array}{*{20}{c}}
 1  &  \ldots   & 1                  \\
\end{array}} ]^{\rm T}_{1 \times L}}.
\end{equation}

\end{proposition}

\begin{Proof}
See Appendix \ref{AppendixProEFIMPPM}.
\end{Proof}

\begin{remark}   \label{RemarkPPMsingular}
The values in the first and $(L+1)$-th columns of the matrix ${{\bf{I}}_{\boldsymbol \eta_{\rm ppm}}}$ are found to be consistent, which
means the FIM $\mathbf{I}_{\boldsymbol\theta_{\rm ppm}}$ is {\it singular}. Consequently, the CRLBs for the estimated parameters in $\boldsymbol \theta_{\rm ppm}$ cannot be directly obtained.
Observe that the first column corresponds to the vector related to the time-delay $\tau_1$, while the $(L+1)$-th column
\vspace{-2pt}
\begin{equation}      \nonumber
\bf{I}_{\rm ppmd} = \big[ \: {{{\bf{H}}^{\rm{T}}}{{\bf{\Lambda }}_{\boldsymbol \tau, \boldsymbol \tau}}{\bf{E}}},  \: {{{\bf{E}}^{\rm{T}}}{{\bf{\Lambda }}_{\boldsymbol \tau, \: {\boldsymbol \tau}}}{\bf{E}}},  \: {\bf 0},  \: {{\bf{\Lambda }}_{{\tilde {\boldsymbol \alpha}}}, \boldsymbol \tau} {\bf{E}}
\: \big]^{\rm T}
\end{equation}
represents the vector associated with the time-shift $\Delta \tau_{\rm q}$ due to data modulation. So the primary cause of this singularity in $\mathbf{I}_{\boldsymbol\theta_{\rm ppm}}$ stems from the coupling between the signal propagation delay $\tau_l$ and the PPM modulation interval $\Delta
\tau_{\rm q}$, as illustrated in \eqref{EqRXPPM}.
In other words, we can only obtain the value of $\tau_l+\tau_{\rm q}$, but cannot determine their individual values.
Therefore, appropriate decoupling solutions are necessary for PPM ISAC scenarios.
\end{remark}

\vspace{-10pt}
\subsection{FIM in ISAC Systems with BPSK Modulation Scheme}

The observation parameters vector in the BPSK case is shown as
\vspace{-2pt}
\begin{equation}\label{EqthetaPPM}
{\boldsymbol {\eta}_{\rm bpsk}}={\Big [{\boldsymbol \tau^{\rm{T}}},{\boldsymbol \phi_{\rm bpsk}^{\rm{T}}},{\tilde {\boldsymbol \alpha}^{\rm{T}}} \Big]^{\rm{T}}},
\end{equation}
where the representation of the phase $\boldsymbol \phi_{\rm bpsk}$ is provided immediately after (\ref{EqRXBPSK}).

\begin{proposition}  \label{PropoEFIMBPSK}

The FIM of estimated parameters vector ${\boldsymbol \theta}_{\rm bpsk}$ in \eqref{Eqthetabpsk} is given by
\begin{equation}     \label{EqEFIMBPSK}
\begin{aligned}
& {{\bf I}_{{\boldsymbol \theta}_{\rm bpsk}}} = {\bf J}_{{\rm{bpsk}}}^{\rm{T}}{{\bf{I}}_{\boldsymbol \eta_{\rm bpsk}}}{{\bf J}_{{\rm{bpsk}}}}      \\
& \!= \!\!\!{\left[\!\!\! {\begin{array}{*{20}{c}}
  {{{\bf{H}}^{\rm{T}}}{{\bf{\Lambda }}_{\boldsymbol \tau, \boldsymbol \tau }}{\bf{H}}}
& {\bf 0}
& {\bf 0}
& {{{\bf{H}}^{\rm{T}}}{{\bf{\Lambda }}_{\boldsymbol \tau, {\tilde {\boldsymbol \alpha}}}}}           \\
  {\bf 0}
& \mathfrak{b}{{{\bf{H}}^{\rm{T}}}{{\bf{\Lambda }}_{\boldsymbol \phi^{0},  \boldsymbol \phi^{0}}}{\bf{H}}}
& \mathfrak{b} {{{\bf{H}}^{\rm{T}}}{{\bf{\Lambda }}_{\boldsymbol \phi^{0}, \boldsymbol \phi^{0}}}{\bf{E}}}
& {\bf 0}           \\
  {\bf 0}
& \mathfrak{b} {{{\bf{E}}^{\rm{T}}}{{\bf{\Lambda }}_{\boldsymbol \phi^{0}, \boldsymbol \phi^{0} }}}{\bf{H}}
& \mathfrak{b} {{{\bf{E}}^{\rm{T}}}{{\bf{\Lambda }}_{\boldsymbol \phi^{0},{\boldsymbol \phi^{0}}}}{\bf{E}}}
& {\bf 0}           \\
  {{\bf{\Lambda }}_{{\tilde {\boldsymbol \alpha}}, \boldsymbol \tau}} {\bf{H}}
& {\bf 0}
& {\bf 0}
& {{\bf{\Lambda }}_{{\tilde {\boldsymbol \alpha}}, {\tilde {\boldsymbol \alpha}}}}
\end{array}}\!\!\! \right]}\!,                 \\
\end{aligned}
\end{equation}
where ${{\bf{I}}_{\boldsymbol \eta_{\rm bpsk}}}$ is the FIM of $\boldsymbol \eta_{\rm bpsk}$, and the numerical values of the elements in
${{\bf I}_{{\boldsymbol \eta}_{\rm bpsk}}}$ are consistent with those in ${{\bf I}_{\boldsymbol \eta_{\rm s}}}$, ${\bf J}_{{\rm{bpsk}}}$ is the Jacobian matrix which is mapping from $\boldsymbol \eta_{\rm bpsk}$ to $\boldsymbol \theta_{\rm bpsk}$.

\end{proposition}

\begin{Proof}
See Appendix \ref{AppendixProEFIMBPSK}.
\end{Proof}

\begin{remark}

A similar phenomenon, as noted in Remark 1, can be observed.
The numerical values in the $(L+1)$-th and $(2L+1)$-th column of the matrix ${{\bf{I}}_{\boldsymbol \eta_{\rm bpsk}}}$ are consistent, indicating that
the FIM $\mathbf{I}_{\boldsymbol\theta_{\rm bpsk}}$ to be {\it singular}. Consequently, the CRLBs for the estimated parameters in $\boldsymbol \theta_{\rm bpsk}$ cannot be directly obtained.
Observe that the $(L+1)$-th column corresponds to the vector associated with the Doppler shift $f_{\rm d1}$, while the $(2L+1)$-th column
\vspace{-5pt}
\begin{equation}  \nonumber
\bf{I}_{\rm bpskd} = \big[ \: {\bf {0}},  \:  \mathfrak{b}{{{\bf{H}}^{\rm{T}}}{{\bf{\Lambda }}_{\boldsymbol \phi^{0},  \boldsymbol \phi^{0}}}{\bf{H}}},  \: \mathfrak{a} {{{\bf{E}}^{\rm{T}}}{{\bf{\Lambda }}_{\boldsymbol \phi^{0}, \boldsymbol \phi^{0} }}}{\bf{H}},  \:  {\bf 0} \: \big]^{\rm T}
\end{equation}
represents the vector realted to the data $\varphi _{\rm{bpsk}}$.
The singularity in $\mathbf{I}_{\boldsymbol\theta_{\rm bpsk}}$ arises from the coupling between the target's Doppler shift
$f_{\rm dl}$ and the phase component $\varphi _{\rm{bpsk}}$ related to the data, as depicted in \eqref{EqRXBPSK}.
In other words, we can only obtain the value of $\varphi _{\rm{bpsk}}+f_{\rm dl}$ without being able to separate their individual values.
Therefore, effective decoupling methods are required for BPSK ISAC scenarios.
\end{remark}

\section{Decoupling Strategy and a Case Study with Resource Allocation}     \label{SecDecoupling}

\subsection{Pilot-based Decoupling}   \label{Sec PilotDecoupPilot}

A widely adopted method for decoupling sensing and communication parameters involves using pilots without embedded data, as illustrated in Fig.~\ref{FigPilotFrameStructure}. Pilots facilitate this decoupling by providing signal transmission delays or the received signal phase without data modulation ({\it pulse position or phase synchronization}). Meanwhile, the start frame delimiter (SFD) indicates the beginning of data transmission \cite{UWB4z}. We introduce a typical frame structure consisting of $P$ pilots and $D$ data, transmitted with UWB pulses. The following section presents the theoretical framework for data demodulation, based on parameter estimation theory.
\begin{figure}
  \centering
  \includegraphics[width = 0.5\columnwidth] {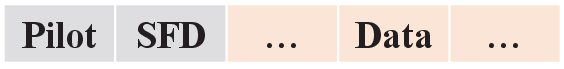}
  \vspace{-5pt}
  \caption{Pilot-based frame structure.}   \label{FigPilotFrameStructure}
  \vspace{-10pt}
\end{figure}

\subsubsection{PPM Case}   \label{SecPPMPilot}
\begin{figure*} [b]
\vspace{-8pt}
\hrulefill
\vspace{-3pt}
\begin{equation}              \label{EqEFIMDiffSeq}
{{\bf{I}}_{{\boldsymbol \varpi_{\rm diff}}}} = \mathbf {P}_{\rm diff}^{\text T} {{\bf{I}}_{\boldsymbol \eta_{\rm diff}}} \mathbf {P}_{\rm diff}
= \left[
\begin{matrix}
  {{\bf{\Lambda }}_{{\boldsymbol t^{\rm ref}}, {\boldsymbol t^{\rm ref}}}} + {{\bf{\Lambda }}_{{\boldsymbol t^0}, {\boldsymbol t^0}}}     \hfill
& {{\bf{\Lambda }}_{{\boldsymbol t^{\rm 0}}, {\boldsymbol t^0}}}       \hfill
& \fcolorbox{red}{white}{${{\bf{\Lambda }}_{{\boldsymbol t^{\rm ref}}, {\boldsymbol t^{\rm ref}}}}$}          \hfill
& {\bf{0}}           \hfill
& {\bf{0}}           \hfill
& {\bf{0}}           \hfill           \\

  {{\bf{\Lambda }}_{{\boldsymbol t^0}, {\boldsymbol t^{\rm 0}}}}      \hfill
& {{\bf{\Lambda }}_{{\boldsymbol t^0}, {\boldsymbol t^0}}}            \hfill
& -{{\bf{\Lambda }}_{{\boldsymbol t^0}, {\boldsymbol t^0}}}           \hfill
& {\bf{0}}         \hfill
& {\bf{0}}         \hfill
& {{{\bf{\Lambda }}_{{\boldsymbol t^{\rm 0}}, {\tilde {\boldsymbol \alpha}} }}}  \hfill     \\

\fcolorbox{red}{white}{${{\bf{\Lambda }}_{{\boldsymbol t^{\rm ref}}, {\boldsymbol t^{\rm ref}}}}$}       \hfill
& {\bf{0}}                \hfill
& {{\bf{\Lambda }}_{{\boldsymbol t^{\rm ref}}, {\boldsymbol t^{\rm ref}}}} + {{\bf{\Lambda }}_{{\boldsymbol t^1}, {\boldsymbol t^1}}}     \hfill
& {{\bf{\Lambda }}_{{\boldsymbol t^{\rm 1}}, {\boldsymbol t^{\rm 1}}}}         \hfill
& {\bf{0}}                \hfill
& {\bf{0}}                \hfill        \\

  {\bf{0}}      \hfill
& {\bf{0}}    \hfill
& {{\bf{\Lambda }}_{{\boldsymbol t^{\rm 1}}, {\boldsymbol t^{\rm 1}}}}           \hfill
& {{\bf{\Lambda }}_{{\boldsymbol t^{\rm 1}}, {\boldsymbol t^{\rm 1}}}}           \hfill
& {\bf{0}}    \hfill
& {{{\bf{\Lambda }}_{{\boldsymbol t^{\rm 1}}, {\tilde {\boldsymbol \alpha}} }}}   \hfill       \\

  {\bf{0}}      \hfill
& {\bf{0}}      \hfill
& {\bf{0}}           \hfill
& {\bf{0}}           \hfill
& {{{\bf{\Lambda }}_{{\boldsymbol \phi}, {\boldsymbol \phi}}}}                   \hfill
& {\bf{0}}         \hfill  \\

  {\bf{0}}       \hfill
& {{{\bf{\Lambda }}_{\tilde {\boldsymbol \alpha}, {\boldsymbol t^{\rm 0}} }}}     \hfill
& {\bf{0}}       \hfill
& {{{\bf{\Lambda }}_{\tilde {\boldsymbol \alpha}, {\boldsymbol t^{1}}}}}         \hfill
& {\bf{0}}      \hfill
& {{{\bf{\Lambda }}_{\tilde {\boldsymbol \alpha}, \tilde{ \boldsymbol \alpha}}}} \hfill       \\
\end{matrix}  \right].
\end{equation}
\vspace{-10pt}
\end{figure*}

In the pilot-based decoupling ISAC UWB system with PPM schemes, the observation parameters vector is given by
\vspace{-10pt}
\begin{equation}           \label{EqPilotParameterPPM}
{\boldsymbol \eta_{\rm ppm,p}} = {\Big [\boldsymbol \tau _{\rm{p,1}}^{\rm{T}},\boldsymbol \tau _{\rm{d,1}}^{\rm{T}}, \boldsymbol \phi _{\rm{p,1}}^{\rm{T}},
\boldsymbol \phi _{\rm{d,1}}^{\rm{T}}, \tilde {\boldsymbol \alpha}_{\rm{p,1}}^{\rm{T}}, \tilde {\boldsymbol \alpha}_{\rm{d,1}}^{\rm{T}} \Big]^{\rm{T}}},
\end{equation}
where $\boldsymbol{\tau}_{\rm p,1}=\boldsymbol{\tau}$ and $\boldsymbol{\tau}_{\rm d,1} = \boldsymbol{\tau}_{\rm ppm}$
represent the time-delay obtained from the pilot portion and the data portion, respectively,
$\boldsymbol \phi _{\rm{p,1}} = {\left[ {\boldsymbol \phi^0, \cdots , \boldsymbol \phi^\kappa , \cdots ,
\boldsymbol \phi^{{P} - 1}} \right]^{\rm{T}}}$  and
$\boldsymbol \phi _{\rm{d,1}} = {\left[ {\boldsymbol \phi^{P}, \cdots, \boldsymbol \phi^{{P+D} - 1}} \right]^{\rm{T}}}$
represent the phase obtained from the pilot portion and the data portion, respectively,
$\tilde {\boldsymbol \alpha}_{\rm{p,1}} = \tilde {\boldsymbol \alpha}_{\rm{b,1}} = \tilde {\boldsymbol{\alpha}}$, and they are terms related to amplitude for the pilot portion and the data portion.

\begin{proposition}     \label{PropoEFIMPPMpilot}

The FIM of ${\boldsymbol \theta}_{\rm ppm}$ with pilot-based decoupling can be expressed as
\begin{equation}        \label{EqEFIMPPMPilot}
\begin{aligned}
&  {{\bf I}_{{\boldsymbol \theta}_{\rm ppm,p}}}
 = {\bf J}_{{\rm{ppm,p}}}^{\rm{T}}{{\bf{I}}_{\boldsymbol \eta_{\rm ppm,p}}}{{\bf J}_{{\rm{ppm,p}}}}      \\
& \!= \!{\left[\!\!\! {\begin{array}{*{20}{c}}
  {{{\bf{H}}^{\rm{T}}}{{\bf{\Lambda }}_{\boldsymbol \tau, \boldsymbol \tau }^{\rm PD}}{\bf{H}}}
&  {{{\bf{H}}^{\rm{T}}}{{\bf{\Lambda }}_{\boldsymbol \tau, \boldsymbol \tau }^{\rm D}}{\bf{E}}}
& {\bf 0}
& {{{\bf{H}}^{\rm{T}}}{{\bf{\Lambda }}_{\boldsymbol \tau, {\tilde {\boldsymbol \alpha}}}^{\rm PD}}}    \\
  {{{\bf{E}}^{\rm{T}}}{{{\bf{\Lambda }}_{\boldsymbol \tau,{\boldsymbol \tau}}^{\rm D}}}{\bf{H}}}
&  {{{\bf{E}}^{\rm{T}}}{{\bf{\Lambda }}_{\boldsymbol \tau,{\boldsymbol \tau}}^{\rm D}}{\bf{E}}}
& {\bf 0}
& {{{\bf{E}}^{\rm{T}}}{{\bf{\Lambda }}_{\boldsymbol \tau, {\tilde {\boldsymbol \alpha}}}^{\rm D}}}   \\
  {\bf 0}
& {\bf 0}
& \mathfrak{b}^{\rm PD}{{{\bf{H}}^{\rm{T}}}{{\bf{\Lambda }}_{\boldsymbol \phi^{0},  \boldsymbol \phi^{0}}}{\bf{H}}}
& {\bf 0}                                                        \\
  {{\bf{\Lambda }}_{{\tilde {\boldsymbol \alpha}}, \boldsymbol \tau}^{\rm PD}} {\bf{H}}
& {{\bf{\Lambda }}_{{\tilde {\boldsymbol \alpha}}, \boldsymbol \tau}^{\rm D}} {\bf{E}}
& {\bf 0}
& {{\bf{\Lambda }}_{{\tilde {\boldsymbol \alpha}}, {\tilde {\boldsymbol \alpha}}}^{\rm PD}}
\end{array}} \!\!\! \right]} \!  ,                                \\
\end{aligned}
\end{equation}
where ${{\bf{I}}_{\boldsymbol \eta_{\rm ppm,p}}}$ is the FIM of $\boldsymbol \eta_{\rm ppm,p}$, $\mathbf {J}_{\rm ppm,p}$ denotes the Jacobian matrix that maps ${\boldsymbol \eta_{\rm ppm,p}}$ to ${{\boldsymbol \theta}_{\rm ppm}}$,
$\mathfrak{b}^{\rm PD}{=} \mathfrak{b}|_{{N_{\text f}} = {{P + D}}} $,
${{\bf{\Lambda }}_{\boldsymbol \tau, \boldsymbol \tau}^{\rm PD}}
{ = }{{\bf{\Lambda }}_{\boldsymbol \tau, \boldsymbol \tau}|_{{N_{\text f}} = {P + D}}}$,
${{\bf{\Lambda }}_{\boldsymbol \tau, \boldsymbol \tau}^{\rm D}}
{ = }{{\bf{\Lambda }}_{\boldsymbol \tau, \boldsymbol \tau}|_{{N_{\text f}} = {D}}}$, they are elements of the FIM evaluated at different number of PRIs, with similar interpretations for other elements in the FIM.

In this case, the numerical values of the first and the $(L+1)$-th columns of the FIM  ${{\bf I}_{{\boldsymbol \theta}_{\rm ppm,p}}}$ are not consistent,
indicating that the FIM is non-singular. This non-singularity allows for the effective decoupling of data information from sensing information.

\end{proposition}

\begin{Proof}
See Appendix \ref{AppendixProEFIMPPMPilot}.
\end{Proof}

\begin{remark}
Compared to \eqref{EqEFIMSense}, \eqref{EqEFIMPPMPilot} introduces some submatrices related to the data (submatrices with $\mathbf{E}$). According to {\it Definition \ref{DefEFIM}}, the presence of these data-related submatrices reduces the values of the information matrix associated with the sensing parameters, which further indicates that the presence of data can result in a degradation of sensing performance.
\end{remark}

Based on \eqref{EqEFIMPPMAp}, the EFIM of the communication parameter $\Delta \tau_{\rm q}$ is given by\footnote{EFIM ${{\bf{I}}_{\Delta \tau_{\rm q}}}$ represents the information content associated with the communication data portion in the proposed ISAC system, including both the communication capacity and the information used for data demodulation.}
\begin{equation}       \label{EqEFIMCommun}
{{\bf{I}}_{\Delta \tau_{\rm q}}} = {\bf A}_{\rm q} - {\bf B}_{\rm q}^{\rm T}{{\bf C}_{\rm q}^{ - 1}}{\bf B}_{\rm q},
\end{equation}
where ${\bf A}_{\rm q} = {{{\bf{E}}^{\rm{T}}}{{\bf{\Lambda }}_{\boldsymbol \tau,{\boldsymbol \tau}}^{\rm D}}{\bf{E}}}$,
${\bf B}_{\rm q} = \big[\: {{{\bf{H}}^{\rm{T}}}{{\bf{\Lambda }}_{\boldsymbol \tau, \:  \boldsymbol \tau}^{\rm D}}{\bf{E}}}, \:  {\bf 0},  \:
{{\bf{\Lambda }}_{{\tilde {\boldsymbol \alpha}}, \:  \boldsymbol \tau}^{\rm D}} {\bf{E}} \:  \big ]^{\text T}$, and
\begin{equation}        \label{EqEFIMCommunCC}
\begin{aligned}
{{\bf C}_{\rm q}}
 = {\left[ {\begin{array}{*{20}{c}}
  {{{\bf{H}}^{\rm{T}}}{{\bf{\Lambda }}_{\boldsymbol \tau, \boldsymbol \tau}^{\rm PD}}{\bf{H}}}
&  {\bf 0}
& {{{\bf{H}}^{\rm{T}}}{{\bf{\Lambda }}_{\boldsymbol \tau}, {\tilde {\boldsymbol \alpha}}}^{\rm PD}}    \\
  {\bf 0}
& \mathfrak{b}^{\rm PD}{{{\bf{H}}^{\rm{T}}}{{\bf{\Lambda }}_{\boldsymbol \phi^{0},  \boldsymbol \phi^{0}}}{\bf{H}}}
& {\bf 0}                         \\
  {{\bf{\Lambda }}_{{\tilde {\boldsymbol \alpha}}, \boldsymbol \tau}^{\rm PD}} {\bf{H}}
& {\bf 0}
& {{\bf{\Lambda }}_{{\tilde {\boldsymbol \alpha}}, {\tilde {\boldsymbol \alpha}}}^{\rm PD}}
\end{array}} \right].}
\end{aligned}
\end{equation}

\subsubsection{BPSK Case}

In the pilot-based decoupling ISAC UWB system with a BPSK modulation scheme, the observation parameters vector can be expressed as
\begin{equation}           \label{EqPilotParameterBPSK}
{\boldsymbol \eta_{\rm bpsk, p}} = {\Big[\boldsymbol \tau _{\rm{p,2}}^{\rm{T}},\boldsymbol \tau _{\rm{d,2}}^{\rm{T}}, \boldsymbol \phi _{\rm{p,2}}^{\rm{T}},
\boldsymbol \phi _{\rm{d,2}}^{\rm{T}}, \tilde {\boldsymbol \alpha}_{\rm{p,2}}^{\rm{T}}, \tilde {\boldsymbol \alpha}_{\rm{d,2}}^{\rm{T}} \Big]^{\rm{T}}},
\end{equation}
where $\boldsymbol{\tau}_{\rm p,2} = \boldsymbol{\tau}_{\rm d,2} = \boldsymbol{\tau}$,
$\boldsymbol{\tau}_{\rm p,2}$ and $\boldsymbol{\tau}_{\rm d,2}$ correspond to the time-delay obtained from the pilot portion and the data portion, respectively,
$\boldsymbol \phi _{\rm{p,2}} = {\big[ {\boldsymbol \phi^0, \cdots , \boldsymbol \phi ^\kappa , \cdots , \boldsymbol \phi^{{P} - 1}} \big]^{\rm{T}}}$  and
$\boldsymbol \phi _{\rm{d,2}} = {\big[ {\boldsymbol \phi _{\rm{bpsk}}^{P}, \cdots, \boldsymbol \phi _{\rm{bpsk}}^{{P+D} - 1}} \big]^{\rm{T}}}$
correspond to the phase observed from the pilot portion and the data portion, respectively,
$\tilde {\boldsymbol \alpha}_{\rm{p,2}} = \tilde {\boldsymbol \alpha}_{\rm{b,2}} = \tilde {\boldsymbol{\alpha}}$, and they are terms related to the amplitude of the received signals.

\begin{proposition}   \label{PropoEFIMBPSKpilot}

The FIM of ${\boldsymbol \theta}_{\rm bpsk}$ with pilot-based decoupling can be expressed as
{\setlength{\arraycolsep}{1pt}
\begin{equation}        \label{EqEFIMBPSKPilot}
\begin{aligned}
&  {{\bf I}_{{\boldsymbol \theta}_{\rm bpsk,p}}}
= {\bf J}_{{\rm{bpsk,p}}}^{\rm{T}}{{\bf{I}}_{\boldsymbol \eta_{\rm bpsk,p}}}{{\bf J}_{{\rm{bpsk,p}}}}      \\
& = \!\!\!  {\left[ {\begin{array}{*{20}{c}}
  {{{\bf{H}}^{\rm{T}}}{{\bf{\Lambda }}_{\boldsymbol \tau, \boldsymbol \tau}^{\rm PD}}{\bf{H}}}
& {\bf 0}
& {\bf 0}
& {{{\bf{H}}^{\rm{T}}}{{\bf{\Lambda }}_{\boldsymbol \tau, {\tilde {\boldsymbol \alpha}}}^{\rm PD}}}    \\
  {\bf 0}
& \mathfrak{b}^{\rm PD}{{{\bf{H}}^{\rm{T}}}{{\bf{\Lambda }}_{\boldsymbol \phi^{0},  \boldsymbol \phi^{0}}}{\bf{H}}}
& \mathfrak{a}^{\rm D} {{{\bf{H}}^{\rm{T}}}{{\bf{\Lambda }}_{\boldsymbol \phi^{0}, \boldsymbol \phi^{0}}}{\bf{E}}}
& {\bf 0}     \\
  {\bf 0}
& \mathfrak{a}^{\rm D} {{{\bf{E}}^{\rm{T}}}{{\bf{\Lambda }}_{\boldsymbol \phi^{0}, \boldsymbol \phi^{0} }}}{\bf{H}}
& \mathfrak{b}^{\rm D} {{{\bf{E}}^{\rm{T}}}{{\bf{\Lambda }}_{\boldsymbol \phi^{0},{\boldsymbol \phi^{0}}}}{\bf{E}}}
& {\bf 0}     \\
  {{\bf{\Lambda }}_{{\tilde {\boldsymbol \alpha}}, \boldsymbol \tau}^{\rm PD}} {\bf{H}}
& {\bf 0}
& {\bf 0}
& {{\bf{\Lambda }}_{{\tilde {\boldsymbol \alpha}}, {\tilde {\boldsymbol \alpha}}}^{\rm PD}}
\end{array}} \right]} \!,
\end{aligned}
\end{equation}}
where ${{\bf{I}}_{\boldsymbol \eta_{\rm bpsk,p}}}$ is the FIM of $\boldsymbol \eta_{\rm bpsk,p}$, $\mathbf {J}_{\rm bpsk,p}$ is the Jacobian matrix which maps the observation parameters vector ${\boldsymbol {\eta}_{\rm bpsk,p}}$ to interested parameters vector $\boldsymbol\theta_{\rm bpsk}$,
\begin{equation}  \nonumber
\begin{aligned}
\mathfrak{a}^{\rm D}& =\sum\limits_{\kappa  = P}^{{P+D} - 1} {-{\rm j}2\pi \kappa {T_{\text f}}}
= -{\rm j} \pi {T_{\text f}}{D}\left( {{2P+D} - 1} \right),   \\
\mathfrak{b}^{\rm D} & = \sum\limits_{\kappa  = P}^{{P+D} - 1} {(-{\rm j}2\pi \kappa {T_{\text f}})} {(-{\rm j}2\pi \kappa {T_{\text f}})}  \\
& = {{2 {\pi}^2 {T_{\text f}^2}} \over 3 }({(P+D-1)}(P+D)(2(P+D)-1)-  \\
& {P(P-1){(2P-1)}}).
\end{aligned}
\end{equation}

\end{proposition}

Here, matrix ${{\bf I}_{{\boldsymbol \theta}_{\rm bpsk,p}}}$ is non-singular, the decoupling of data information from sensing information can be achieved.

\begin{Proof}
See Appendix \ref{AppendixProEFIMBPSKPilot}.
\end{Proof}

\vspace{-12pt}

\subsection{Differential-based Decoupling in the PPM Case}     \label{SecDecouplingDiff}

Based on the PPM modulation principle, the data carried by a pulse can be determined from the time delay differences between the data pulse and SFD, as shown in Fig.~\ref{FigDifferential-based decoupling strategy}. We set the time $t_{1}^{\rm ref}$ as the starting symbol position, which is typically determined by the SFD.
This section presents {\it a differential decoupling strategy without pilots} and develops the theoretical framework for sensing parameter estimation.

\begin{figure}
  \centering
  \includegraphics[width = 0.8\columnwidth]{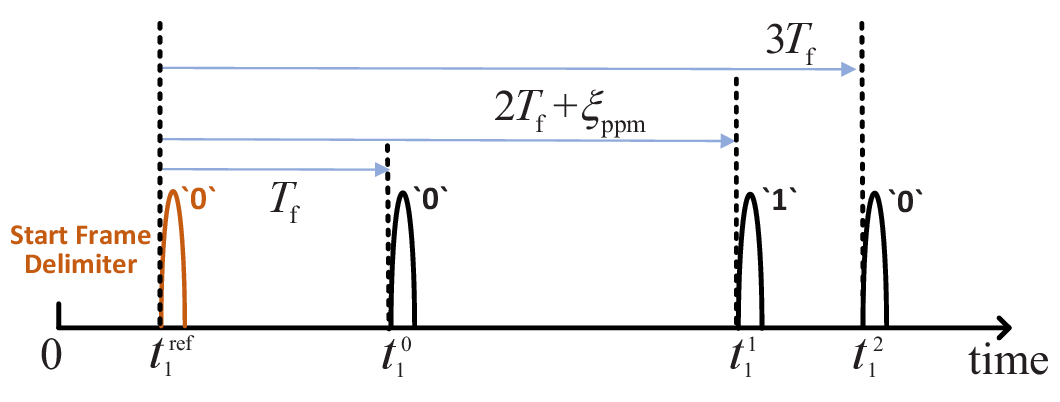}
  \caption{Pulse intervals of the first path in PPM modulation\big(${\xi _{{\rm{ppm}}}} = {T_{\text f}/2}$ \big).}   \label{FigDifferential-based decoupling strategy}
  \vspace{-12pt}
\end{figure}

Firstly, we establish an intermediate parameter vector based on the time differences of pulse arrivals as \eqref{EqDiffSeqVector}.
The data bit can be obtained by the time difference between signal arrival time, such as ${{\boldsymbol t}^0} - {{\boldsymbol t}^{\rm ref}}$, thereby enabling data demodulation.
\begin{equation}         \label{EqDiffSeqVector}
\boldsymbol \varpi_{\rm diff} =
\big[
{{\boldsymbol t}^0} - {{\boldsymbol t}^{\rm ref}},
{{\boldsymbol t}^0}, ...,
{{\boldsymbol t}^{{N_f} - 1}} - {{\boldsymbol t}^{\rm ref}},
{{\boldsymbol t}^{{N_f} - 1}},
{{\boldsymbol \phi}},
{\tilde {\boldsymbol \alpha}}  \big],
\end{equation}
where ${{\boldsymbol t}^{N_{\text f}-1}} = \big[t_1^{N_{\text f}-1}, ..., t_L^{N_{\text f}-1} \big]$, $t_1^{N_{\text f}-1} = \tau_1 + \Delta \tau_{\text q}+(N_{\text f} - 1)T_{\text f}$ is the absolute time of the first path in the ${(N_{\text f}-1)}$-th PRI, ${{\boldsymbol t}^{\rm ref}} = \big[t^{\rm ref}_1, ..., t^{\rm ref}_L \big]$ represents the signal arrival
times of SFD.

Unlike \eqref{EqthetaPPM}, the observation parameters vector should be updated to \eqref{EqObserDiff} accordingly.
\begin{equation}       \label{EqObserDiff}
{\boldsymbol {\eta}_{\rm diff}}={\big[ {\boldsymbol t ^{\rm{T}}},{\boldsymbol \phi}^{\rm{T}},{\tilde {\boldsymbol \alpha}^{\rm{T}}} \big]^{\rm{T}}},
\end{equation}
where $\boldsymbol t = \big[\boldsymbol t^{\rm ref}, \boldsymbol t^{0}, ..., \boldsymbol t^{N_{\text f}-1} \big]$.

\begin{proposition}   \label{PropoEFIMDiff2}
We denote the two Jacobian matrices for transforming from ${\boldsymbol {\eta}_{\rm diff}}$ to $\boldsymbol \varpi_{\rm diff}$ and from $\boldsymbol \varpi_{\rm diff}$ to ${\boldsymbol \theta_{{\rm ppm}}}$ as ${\bf P_{{\rm diff}}}$ and ${\bf J_{{\rm diff}}}$, respectively. Then the FIM of ${\boldsymbol \theta_{{\rm ppm}}}$ under the differential process can be expressed as
\begin{equation}      \label{EqEFIMDiff1}
\begin{aligned}
& {{\bf I}_{{\boldsymbol \theta}_{\rm ppm,d}}} =
{\bf J}_{{\rm{diff}}}^{\rm{T}}{\bf P}_{{\rm{diff}}}^{\rm{T}}{{\bf{I}}_{\boldsymbol \eta_{\rm diff}}}{\bf P}_{{\rm{diff}}}{{\bf J}_{{\rm{diff}}}}      \\
& \! = \!{\left[\!\!\! {\begin{array}{*{20}{c}}
  {{{\bf{H}}^{\rm{T}}}{{\bf{\Lambda }}_{\boldsymbol \tau, \boldsymbol \tau }}{\bf{H}}}
& {{{\bf{H}}^{\rm{T}}}{{\bf{\Lambda }}_{\boldsymbol \tau, \boldsymbol \tau }^{2N_{\text f}}}{\bf{E}}}
& {\bf 0}
& {{{\bf{H}}^{\rm{T}}}{{\bf{\Lambda }}_{\boldsymbol \tau, {\tilde {\boldsymbol \alpha}}}}}    \\
  {{{\bf{E}}^{\rm{T}}}{{{\bf{\Lambda }}_{\boldsymbol \tau,{\boldsymbol \tau}}^{2N_{\text f}}}}{\bf{H}}}
& {{{\bf{E}}^{\rm{T}}}{{\bf{\Lambda }}_{\boldsymbol \tau,{\boldsymbol \tau}}^{5N_{\text f}}}{\bf{E}}}
& {\bf 0}
& {{{\bf{E}}^{\rm{T}}}{{\bf{\Lambda }}_{\boldsymbol \tau, {\tilde {\boldsymbol \alpha}}}}}   \\
  {\bf 0}
& {\bf 0}
& \mathfrak{b}{{{\bf{H}}^{\rm{T}}}{{\bf{\Lambda }}_{\boldsymbol \phi^{0},  \boldsymbol \phi^{0}}}{\bf{H}}}
& {\bf 0}                                                \\
  {{\bf{\Lambda }}_{{\tilde {\boldsymbol \alpha}}, \boldsymbol \tau}} {\bf{H}}
& {{\bf{\Lambda }}_{{\tilde {\boldsymbol \alpha}}, \boldsymbol \tau}} {\bf{E}}
& {\bf 0}
& {{\bf{\Lambda }}_{{\tilde {\boldsymbol \alpha}}, {\tilde {\boldsymbol \alpha}}}}
\end{array}} \!\!\! \right]} \!  ,                       \\
\end{aligned}
\end{equation}
where ${{\bf{I}}_{\boldsymbol \eta_{\rm diff}}}$ is the FIM of ${\boldsymbol {\eta}_{\rm diff}}$, similar to ${{\bf{I}}_{\boldsymbol \eta_{\rm s}}}$.
Consequently, ${{\bf I}_{{\boldsymbol \theta}_{\rm ppm,p}}}$ becomes non-singular, which allows for the effective decoupling of sensing and communication parameters.
\end{proposition}

{\bf Specially,} we denote the FIM of intermediate parameter vector $\boldsymbol \varpi_{\rm diff}$ as ${{\bf{I}}_{{\boldsymbol \varpi_{\rm diff}}}}$, which is written as \eqref{EqEFIMDiffSeq}.{\footnote {For simplicity, we use two pulses as an example.}}
Among the diagonal elements, $[\: {{\bf{\Lambda }}_{{\boldsymbol t^{\rm ref}}, {\boldsymbol t^{\rm ref}}}} + {{\bf{\Lambda }}_{{\boldsymbol t^0}, {\boldsymbol t^0}}} \:]$ represents the data-related part, and this indicates that the variance introduced by the communication data is also doubled compared to the pilot-based decoupling strategy, which is consistent with the principles of TDOA-based positioning systems \cite{TDOA}. Due to the coupling relationship between the data and delay, an increase in the variance of data estimation can also lead to a decline in the performance of delay estimation.
Meanwhile ${{\bf{\Lambda }}_{{\boldsymbol t^0}, {\boldsymbol t^0}}}$ represents the delay-related part. In this paper, we assume that the differential demodulation processes between different data are independent. Therefore, the elements within the red rectangular box (representing the correlation between data) should be set to zero.

\begin{Proof}
For the detailed derivation process, see Appendix~\ref{AppendixProEFIMDiff2}.
\end{Proof}

\begin{remark}
Because differential demodulation depends on pulse intervals, this method remains applicable in cases {\it where there is an asynchronous relationship between the transmitter and receiver clocks} (i.e., an initial clock offset exists, and synchronization of the received signal is not required). Accordingly, the system's sensing parameters is composed of both the Doppler shift and a pseudo-delay (the sum of the clock offset and the signal propagation delay).
\end{remark}

\section{Numerical Results And Discussions}

In this section, we illustrate the applications of our analytical results through numerical examples. We intentionally focus on a simple network to gain insights, even though our analytical results are applicable to arbitrary topologies with any number of MPCs and pulses in the received waveforms. Our
numerical results will be applicable to the performance analysis of a practical UWB-ISAC system, as the parameters specified in the IEEE UWB standards-802.15.4a and 802.15.4z are used for our numerical calculation. The system parameters are detailed in Table \ref{parameters1}.
\begin{table} [h]
\vspace{-5pt}
	\centering
	\setlength{\belowcaptionskip}{5pt}
	\caption{System Parameters}
	\label{parameters1}
	\scalebox{0.9}{
		\begin{tabular}{lll}
			\toprule
			Parameters              & Description                             & Value      \\
			\midrule
            $\alpha$                & Temporal spreading factor (ns)          & 0.2       \\
			$L$                     & Number of ISAC channel path             & $3$        \\
			${f_\text s}$           & Sampling frequency (\text GHz)          & $10$       \\
			$T_{\text f}$                   & PRI (\text ns)                          & $100$      \\
            $f_{\text c}$           & Carrier frequency (\text MHz)           & $3993.6$    \\
			\bottomrule
	\end{tabular}}
\end{table}
\begin{remark}     \label{RemarkUWBParameter}
In 2002, the U.S. Federal Communications Commission (FCC) enacted stringent regulations for UWB radios, requiring a frequency bandwidth $\ge$ 500 MHz and a transmit power spectrum density $<$ -41.3 dBm/MHz \cite{biaozhun}, as well as limiting the energy to 37 nJ for per 1 ms duration. The system resource allocation process in this paper also follows this principle.
\end{remark}

\vspace{-10pt}

\subsection{Tradeoff between communication and sensing}

Fig.~\ref{ResultData_Rate} illustrates the system performance using a receiver operating characteristic curve.
It depicts the relationship between range estimation accuracy and communication rate under the pilot-based decoupling framework.\footnote {For clarity, the CRLB presented in this section refers to the square root of the CRLB expression in \eqref{EqCRLBRange}.}
The boundary of the CRLB-rate region represents the Pareto front, which reflects the optimal trade-off between localization and communication performance~\cite{XiongYF}. As the $SNR$ increases, this boundary extends in the direction of the arrow, indicating simultaneous improvements in both functionalities. In addition, higher $SNR$ levels enable the system to reach optimal performance with a reduced pilot overhead.

\begin{figure}[h]
\vspace{-10pt}
	\centering
	{\includegraphics[width = 0.7\columnwidth]{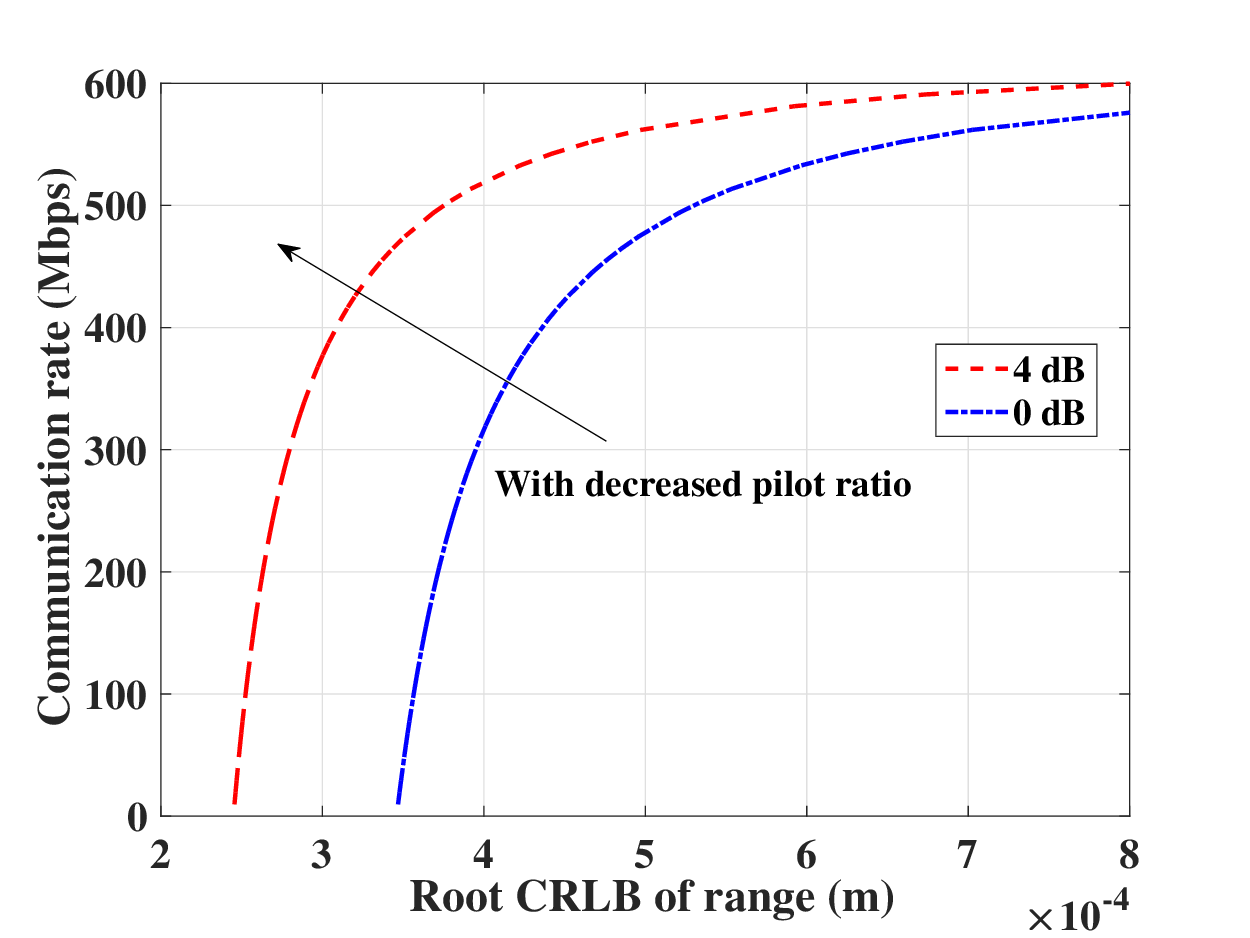}
\vspace{-5pt}
    \caption{The relationship between range estimation accuracy and communication rate. (As the communication rate increases in the direction of the arrow, the required pilot ratio decreases accordingly.).}     \label{ResultData_Rate}}
    \vspace{-2pt}
\end{figure}

\vspace{-10pt}
\subsection{Impact of Data on Sensing Performance}

The integrated pulses after data decoupling can be used to estimate sensing parameters. Different data modulation and demodulation schemes in the ISAC system can influence the sensing performance. We analyze this impact using the derived CRLBs based on \eqref{EqCRLBRange} and provide numerical examples to illustrate the properties in Fig.~\ref{FigResultDemoScheme} and Fig.~\ref{FigResultDemoSchemeDoppler}. We ensure the total number of pulses remains consistent across the different schemes. The observations are given as follows:

\begin{itemize}
\item The BPSK modulation scheme does not impact the estimation of signal transmission delay in the proposed ISAC system, resulting in a ranging performance identical to the sensing-only scenario. Similarly, the PPM scheme does not affect the estimation of the Doppler shift (phase).
\item Compared to a sensing-only system, the ISAC system based on PPM modulation leads to a decline in ranging performance, while the ISAC system using BPSK modulation decreases Doppler estimation performance.
\item A higher number of pulses is required for Doppler estimation than for ranging in order to achieve reasonable lower bounds on estimation performance, such as Doppler estimation around $1~{\rm Hz}$ and ranging around $10^{-4}~{\rm m}$.
    Apart from the energy accumulation effect of multiple pulses, this is mainly due to the strong positive correlation between the resolution of Doppler estimation and the number of pulses. An expression for this relationship is given by $\Delta f_{\rm d} = {1 \over {N_{\text f} T_{\text f}}}$, where $\Delta f_{\rm d}$ represents the Doppler resolution, $N_{\text f}$ and $T_{\text f}$ denote the number of total pulses and the signal transmission period, respectively.
\item The ranging performance achieved with differential decoupling strategies is worse than that of the pilot-based decoupling strategy. On the one hand, regarding the pilot-based decoupling strategy, as mentioned in Sec.~\ref{SecPPMPilot}, pilots without embedded data can provide more accurate ranging information compared to integrated pulses. On the other hand, regarding the differential decoupling strategy, as discussed in Sec.~\ref{SecDecouplingDiff}, the differential process of pulses leads to an increase in data demodulation errors. Due to the coupling between delay and PPM time-shift, this results in a decline in delay estimation (ranging) performance.
\end{itemize}

\begin{figure}[t]
\vspace{-6pt}
	\centering
	{\includegraphics[width = 0.7\columnwidth]{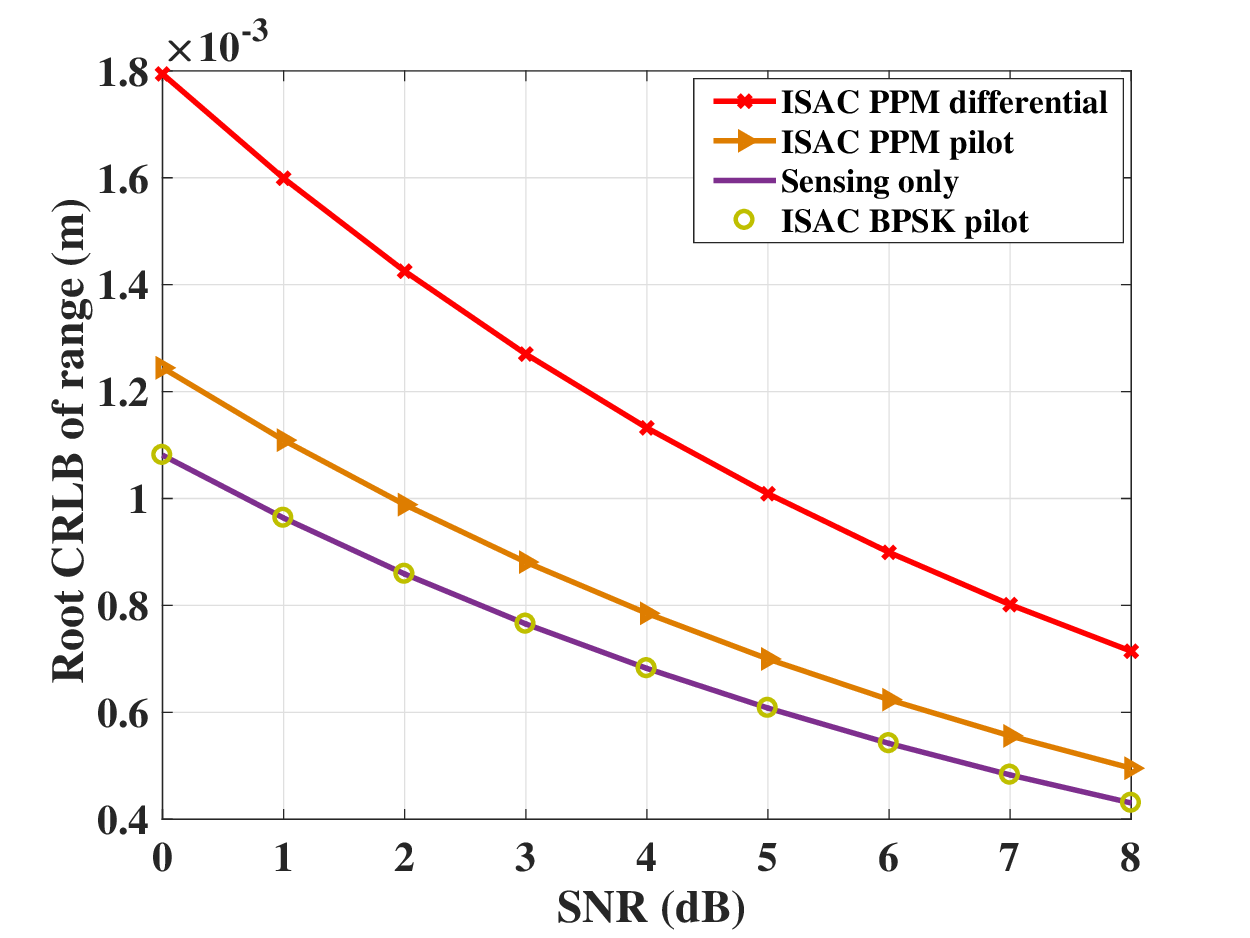}
    \vspace{-5pt}
    \caption {Ranging performance \big(root ${\cal C}(d_1)$\big) with different modulation schemes (Sensing only: 8 pulses with no data; ISAC PPM pilot/ISAC BPSK pilot: 4 pilot pulses + 4 data pulses; ISAC PPM differential: 8 data pulses).}     \label{FigResultDemoScheme}}
    \vspace{-10pt}
\end{figure}

\begin{figure}[t]
	\centering
	{\includegraphics[width = 0.7\columnwidth]{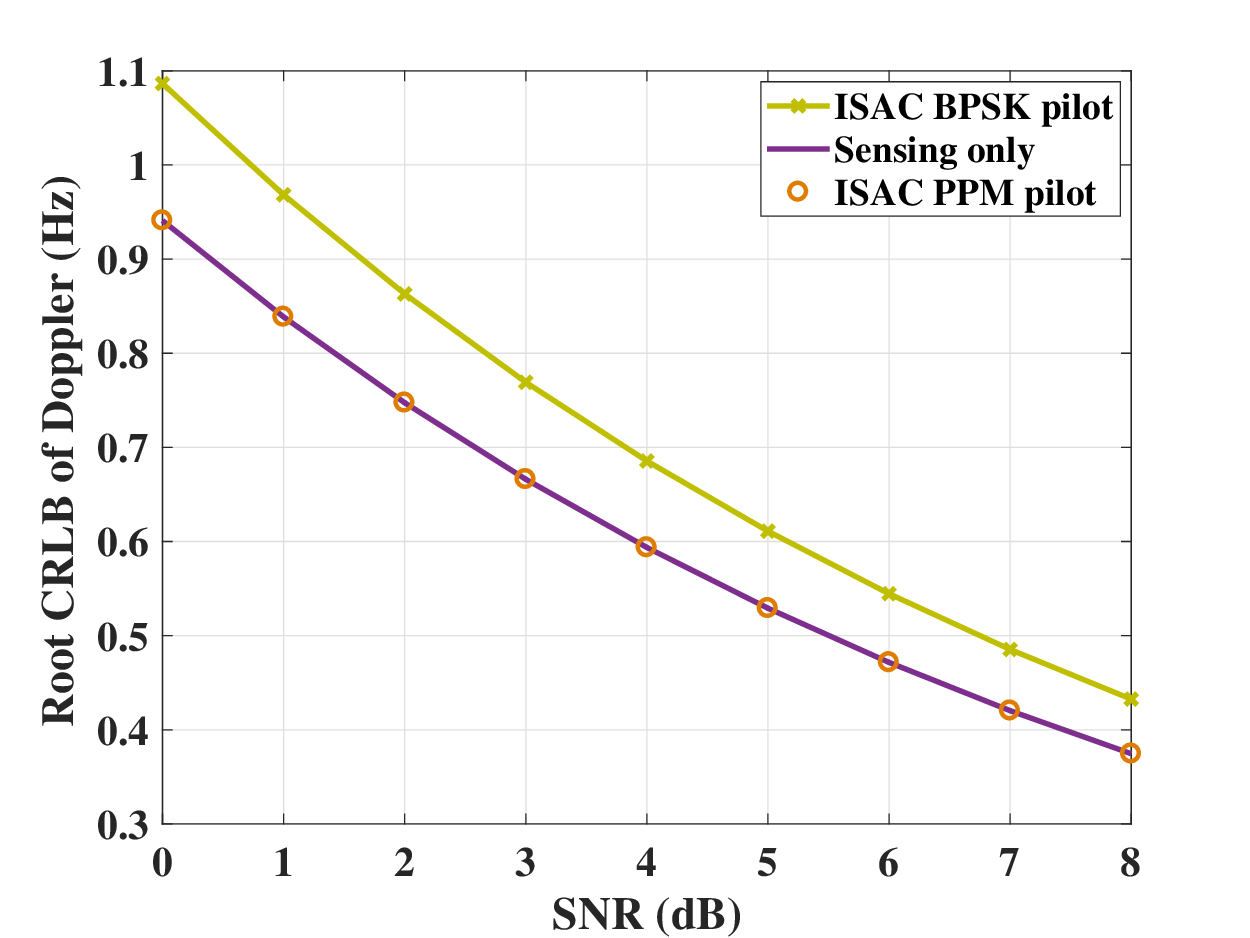}
\vspace{-5pt}
    \caption {Doppler estimation performance with different modulation schemes (Sensing only:2048 pulses with no data; ISAC PPM pilot/ISAC BPSK pilot: 1024 pilot pulses + 1024 data pulses).}     \label{FigResultDemoSchemeDoppler}}
    \vspace{-10pt}
\end{figure}

\subsection{Data Assistance for Sensing in ISAC Systems}

In the proposed pilot-based ISAC system, the data component can further enhance sensing performance compared to using pilots alone, as illustrated in Fig.~\ref{ResultDataAssitRange} and Fig.~\ref{ResultDataAssitDoppler}.

\begin{itemize}
\item We set the number of pilots and data to be equal in the pilot-based ISAC system. The dashed line in Fig.~\ref{ResultDataAssitRange} represents ranging using only pilots. The two bottom lines represent that, in addition to pilots, the data portion is also used to assist with ranging
    under the BPSK and PPM modulation schemes. This indicates that the data portion helps improve ranging performance. We observe that the ranging performance is superior with the BPSK modulation scheme compared to the PPM scheme, due to the coupling between delay and PPM time-shift.
\item If all 2048 pulses are used for data transmission without any pilot signals, as portrayed in the topmost line of Fig.~\ref{ResultDataAssitRange}, the ranging performance with the differential decoupling scheme is the worst, even inferior to that with 1024 pilots. Nevertheless, this configuration achieves a higher communication transmission rate.
\item Fig.~\ref{ResultDataAssitDoppler} demonstrates the supportive role of the data portion in Doppler estimation. It shows that the enhancement of Doppler estimation performance contributed by the data portion. We observe that the Doppler estimation performance is superior with the PPM modulation scheme compared to the BPSK scheme, due to the coupling between signal phase and BPSK modulation phase-shift.
\end{itemize}

\begin{figure}[h]
\vspace{-10pt}
	\centering
	{\includegraphics[width = 0.7\columnwidth]{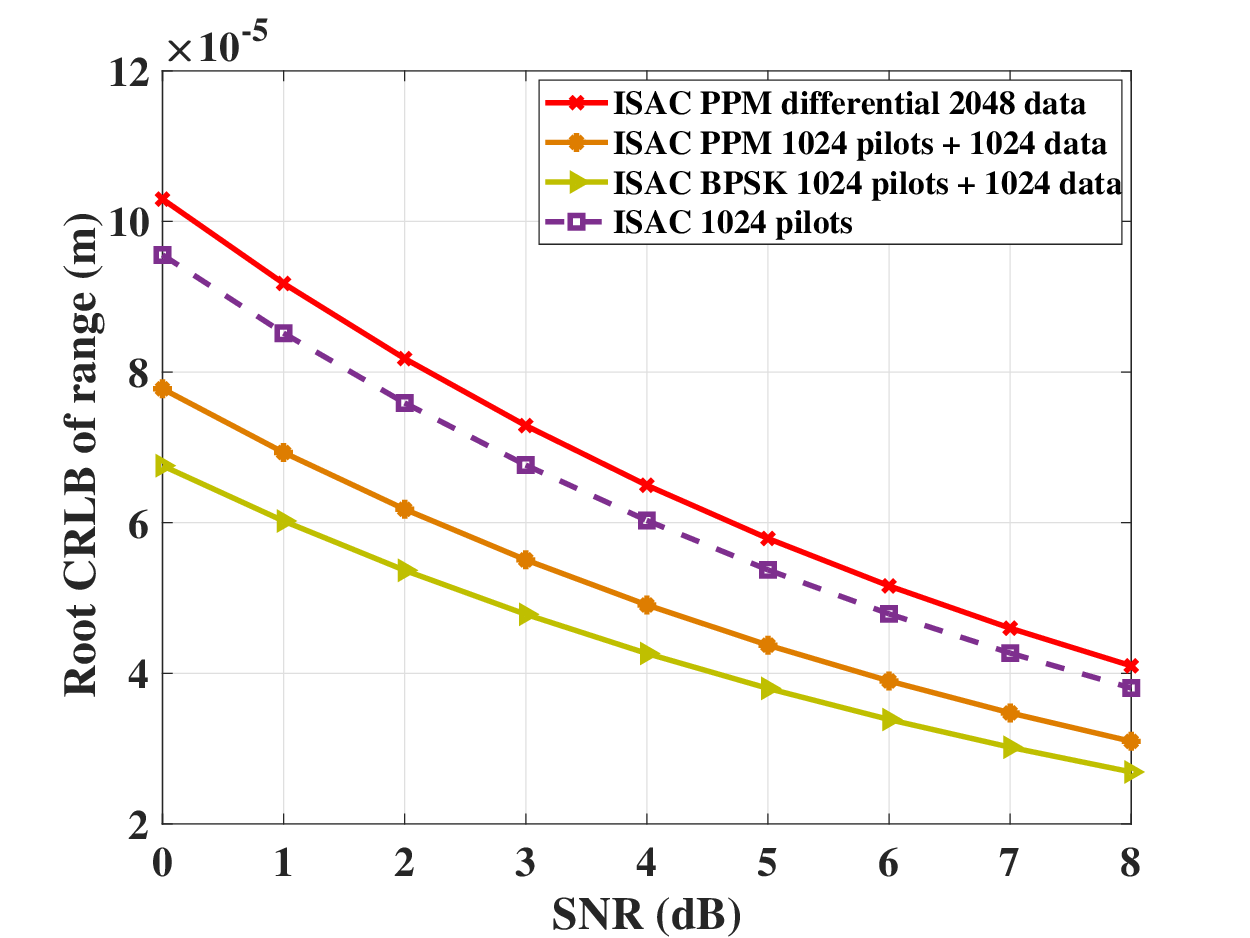}
\vspace{-9pt}
    \caption {Data-assisted range estimation.}     \label{ResultDataAssitRange}}
\vspace{-5pt}
\end{figure}

\begin{figure}[h]
	\centering
	{\includegraphics[width = 0.7\columnwidth]{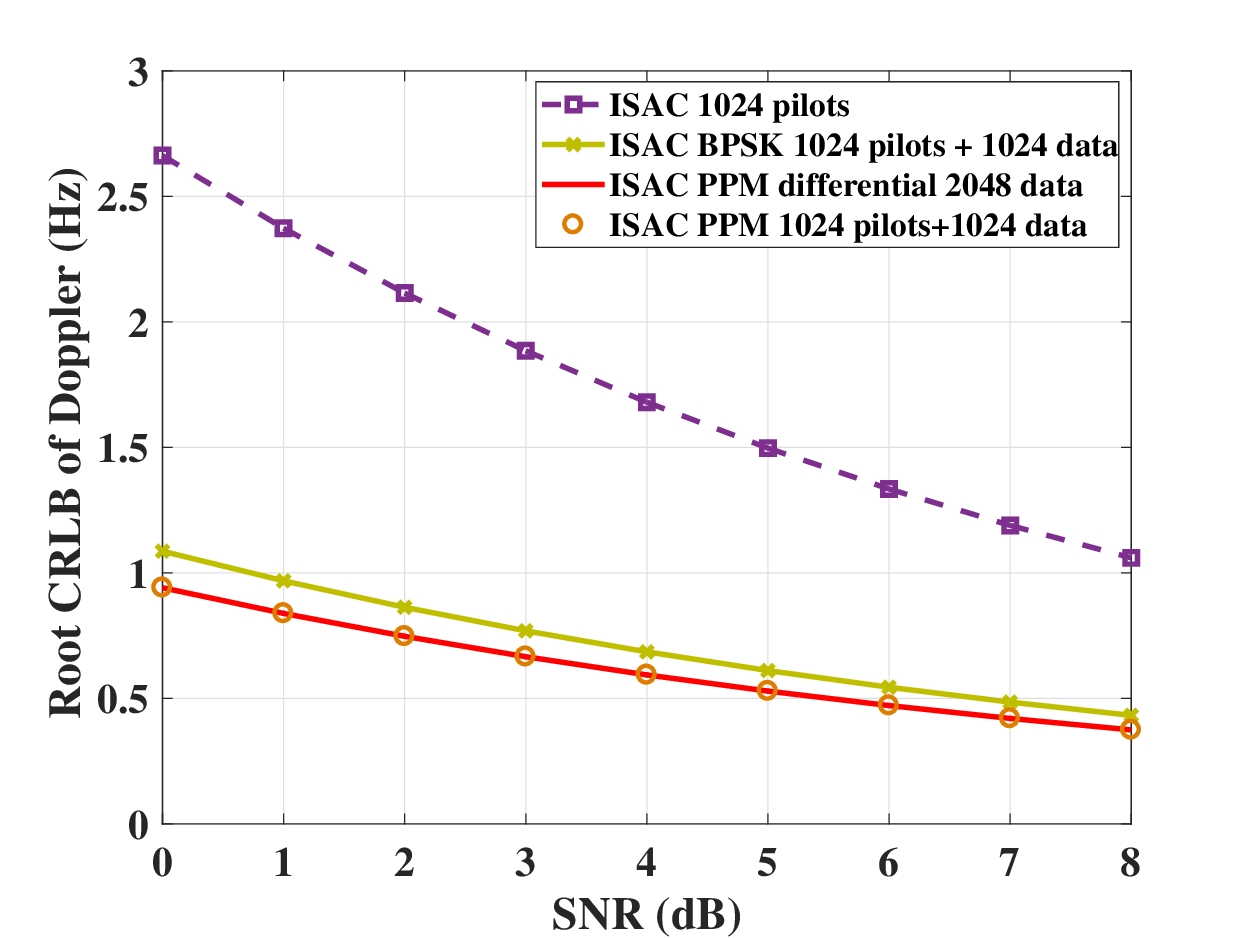}
\vspace{-5pt}
    \caption {Data-assisted Doppler estimation.}     \label{ResultDataAssitDoppler}}
\vspace{-5pt}
\end{figure}

In Fig.~\ref{ResultDemoSchemeNumberbits}, we present the trend of ranging performance as the number of data pulses increases. The sensing-only system has the same number of pulses as the ISAC system. We can see that, for fewer than 22 pulses, the ranging performance of the pilot-based decoupling ISAC system (with a fixed number of pilots) is better than that of the differential decoupling ISAC system. However, when the number of pulses exceeds 22, the opposite holds true.

\begin{figure}[h]
\vspace{-6pt}
	\centering
	{\includegraphics[width = 0.7\columnwidth]{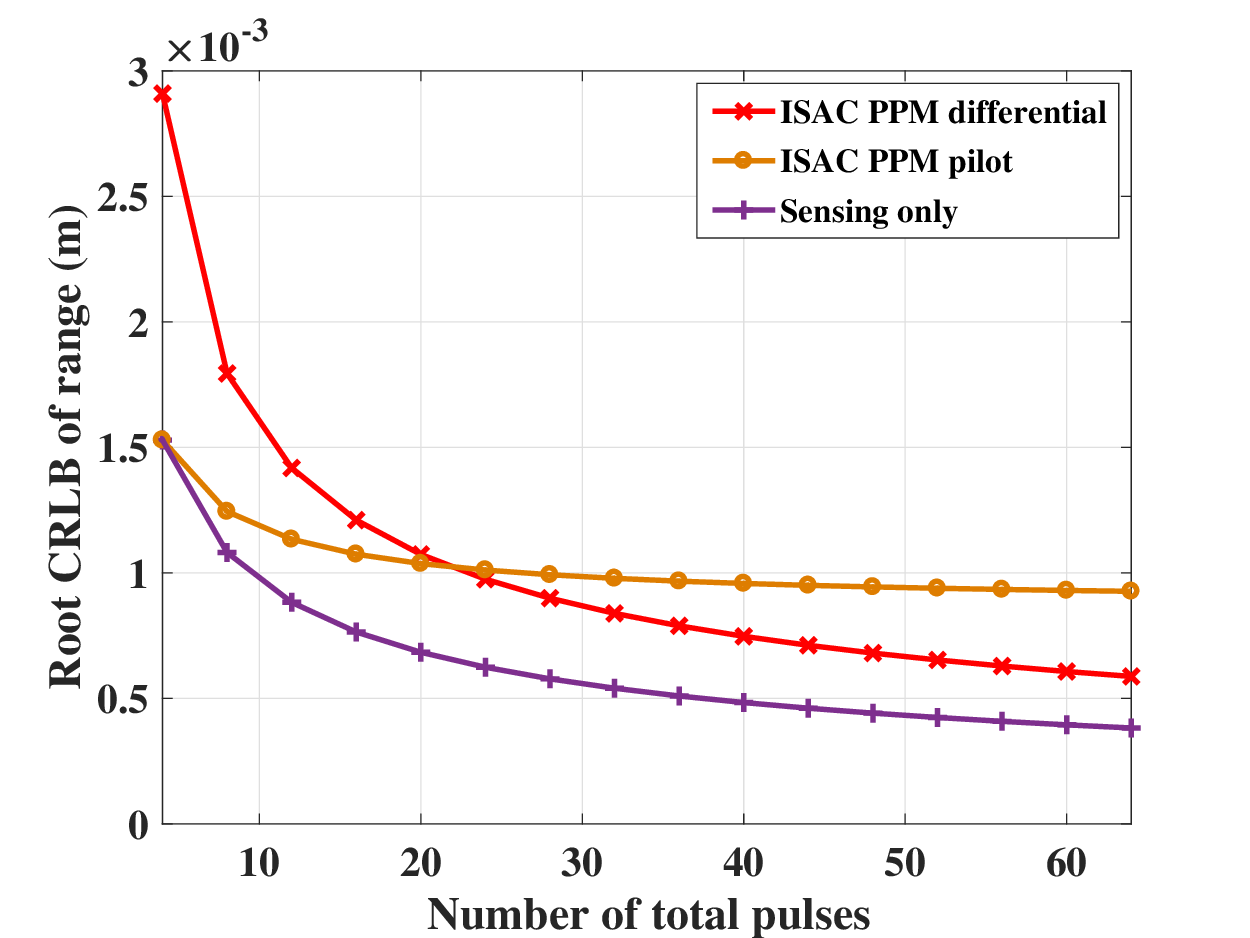}
\vspace{-5pt}
    \caption {Effect of data quantity on ranging performance. (`ISAC PPM pilot' with 4 pilots fixed.)}     \label{ResultDemoSchemeNumberbits}}
\end{figure}

\section{Conclusions}

In this paper, we have developed a theoretical framework for integrating sensing and communication within the UWB system using parameter estimation theory. Based on the FIM and CRLB, the proposed framework characterizes the tight coupling between sensing and communication, and the resulting performance degradation in ISAC systems. We have demonstrated that introducing additional pilots or using differential detection are effective strategies for decoupling, thus making the FIM non-singular. Numerical results indicate that:
(i) Compared to a sensing-only system, data bits can adversely affect sensing performance in the proposed UWB ISAC systems. However, in the pilot-based decoupling ISAC system, data bits can improve sensing performance compared to using only pilots for sensing.
(ii) PPM modulation primarily impacts range estimation performance, whereas BPSK modulation mainly affects Doppler estimation performance.
(iii) When the number of pilot pulses is set to be equal to the number of data pulses, the pilot-based demodulation strategy has a minor impact on range sensing performance compared to the differential-based demodulation strategy. However, when the number of pilots is fixed (set to 4 fixed pilots in
this paper), as the number of data pulses increases, the enhancement effect of the data on sensing performance becomes more significant under the differential demodulation method. Based on the observations above, UWB is a promising candidate waveform for ISAC systems, and specific data signal processing mechanisms may be better suited to specific scenarios.

\vspace{-5pt}
\appendix
\vspace{-2pt}

\subsection{Proof of Corollary \ref{CorolEFIMSense}}   \label{AppendixProEFIMSense}

Firstly, the Jacobian matrix $\mathbf {J}_{\rm s}$ that transforms the observation parameters vector ${\boldsymbol {\eta}_{\rm s}}$ to the estimated parameters vector $\boldsymbol\theta_{\rm s}$  can be written as
\begin{equation}        \label{EqJacSening}
\begin{aligned}
{\bf J_{{\rm s}}}  & = {{\partial {{\boldsymbol \eta }}_{\rm{s}}} \over {\partial {\boldsymbol \theta}_{\rm{s}}}}
    = {{\partial \left( {{\boldsymbol \tau},\boldsymbol \phi ,{{\tilde {\boldsymbol \alpha}}}} \right)} \over
      {\partial \left( {{\tau _{\rm{1}}}, \Delta {\boldsymbol \tau ^{\rm{T}}}, {f_{{\rm{d1}}}},
      {\Delta \boldsymbol f_{\rm{d}}}^{\rm{T}},\tilde {\boldsymbol \alpha}^{\rm{T}}} \right)}}  \\
 & = \left[ \begin{array}{cc:cc:c}
  {{{\partial {\boldsymbol \tau}} \over {\partial {\tau _1}}}}
& {{{\partial {\boldsymbol \tau}}} \over {\partial \Delta {\boldsymbol \tau} ^{\rm{T}}}}
& {{{\partial {\boldsymbol \tau}} \over {\partial {f_{{\rm{d1}}}}}}}
& {{{\partial {\boldsymbol \tau}} \over {\partial \Delta {\boldsymbol f}_{\rm{d}}^{\rm{T}}}}}
& {{{\partial {\boldsymbol \tau}} \over {\partial \tilde {\boldsymbol \alpha}^{\rm{T}}}}}  \\
\hdashline
  {{{\partial {\boldsymbol \phi ^0}} \over {\partial {\tau _1}}}}
& {{{\partial {\boldsymbol \phi ^0}} \over {\partial \Delta {\boldsymbol \tau ^{\rm{T}}}}}}
& {{{\partial {\boldsymbol \phi ^0}} \over {\partial {f_{{\rm{d1}}}}}}}
& {{{\partial {\boldsymbol \phi ^0}} \over {\partial \Delta {\boldsymbol f}_{\rm{d}}^{\rm{T}}}}}
& {{{\partial {\boldsymbol \phi ^0}} \over {\partial \tilde {\boldsymbol \alpha}^{\rm{T}}}}}  \\
  \vdots  &  \vdots  &  \vdots  &  \vdots  &  \vdots          \\
  {{{\partial {\boldsymbol \phi ^{{N_{\text f}} - 1}}} \over {\partial {\tau _1}}}}
& {{{\partial {\boldsymbol \phi ^{{N_{\text f}} - 1}}} \over {\partial \Delta {\boldsymbol \tau ^{\rm{T}}}}}}
& {{{\partial {\boldsymbol \phi ^{{N_{\text f}} - 1}}} \over {\partial {f_{{\rm{d1}}}}}}}
& {{{\partial {\boldsymbol \phi ^{{N_{\text f}} - 1}}} \over {\partial \Delta {\boldsymbol f}_{\rm{d}}^{\rm{T}}}}}
& {{{\partial {\boldsymbol \phi ^{{N_{\text f}} - 1}}} \over {\partial \tilde {\boldsymbol \alpha}^{\rm{T}}}}}  \\
\hdashline
  {{{\partial {{\tilde {\boldsymbol \alpha}}}} \over {\partial {\tau _1}}}}
& {{{\partial {{\tilde {\boldsymbol \alpha}}}} \over {\partial \Delta {\boldsymbol \tau ^{\rm{T}}}}}}
& {{{\partial {{\tilde {\boldsymbol \alpha}}}} \over {\partial {f_{{\rm{d1}}}}}}}
& {{{\partial {{\tilde {\boldsymbol \alpha}}}} \over {\partial \Delta {\boldsymbol f}_{\rm{d}}^{\rm{T}}}}}
& {{{\partial {{\tilde {\boldsymbol \alpha}}}} \over {\partial \tilde {\boldsymbol \alpha}^{\rm{T}}}}}  \\
\end{array} \right]             \\
&  = {\left[ \begin{array}{*{20}{l}}
   {\bf{H}} & {\bf{0}}       & {\bf{0}}    \cr
   {\bf{0}} & {{\bf{L}}_{\rm N}}   & {\bf{0}}    \cr
   {\bf{0}} & {\bf{0}}       & {\bf{I}}    \cr
 \end{array} \right]_{\left( {{N_{\text f}} + 2} \right)L \times 3L}}, \cr
\end{aligned}
\end{equation}
where ${\bf{I}}={{{\partial {{\tilde {\boldsymbol \alpha}}}} \over {\partial \tilde {\boldsymbol \alpha}^{\rm{T}}}}}$ is the identity matrix of dimension $L$, ${\bf{L}_{\rm N}} = \big[\: {{\bf{L}}_0}, \: \ldots, \: {{\bf{L}}_{N_{\text f}-1}} \: \big]^{\rm T}$,
\begin{equation}
\begin{aligned}
{\bf{H}} = \Big[ \begin{array}{cccc}
  {{{\partial {\boldsymbol \tau}} \over {\partial {\tau _1}}}}
& {{{\partial {\boldsymbol \tau}} \over {\partial \Delta {\boldsymbol \tau} ^{\rm{T}}}}}
\end{array} \Big],  \quad
\end{aligned}
\end{equation}
\begin{equation}
\begin{aligned}
{{\bf L}_{N_{\text f}-1}}
= \Big[ \! \begin{array}{cccc}
  {{{\partial {\boldsymbol \phi ^{{N_{\text f}} - 1}}} \over {\partial {f_{{\rm{d1}}}}}}}
& {{{\partial {\boldsymbol \phi ^{{N_{\text f}} - 1}}} \over {\partial \Delta {\boldsymbol f}_{\rm{d}}^{\rm{T}}}}}
\end{array} \!\Big]
= 2\pi \left( {{N_{\text f}} - 1} \right){T_{\text f}}{\bf{H}}.
\end{aligned}
\end{equation}

The FIM similar to \eqref{EqIeta} of the observation parameters vector ${\boldsymbol {\eta}_{\rm s}}$ can be expressed as following matrix partition:
{\setlength{\arraycolsep}{1.8pt}
\begin{equation}    \label{EqFIMSensingparemeter2}
{{\bf{I}}_{\boldsymbol \eta_{\rm s}}}
\! = \!\left[
\begin{matrix}
{{{\bf{\Lambda }}_{{\boldsymbol \tau}, {\boldsymbol \tau}}}} \hfill
& {{{\bf{\Lambda }}_{{\boldsymbol \tau}, {\boldsymbol \phi ^0}}}} \hfill
& \cdots  \hfill
& {{{\bf{\Lambda }}_{{\boldsymbol \tau}, {\boldsymbol \phi ^{{N_f} - 1}}}}} \hfill
& {{{\bf{\Lambda}}_{\boldsymbol \tau, \tilde{\boldsymbol\alpha} }}} \hfill  \\

{{{\bf{\Lambda }}_{{\boldsymbol \phi ^0}, \boldsymbol \tau }}} \hfill
& {{{\bf{\Lambda }}_{{\boldsymbol \phi ^0},{\boldsymbol \phi ^0}}}} \hfill
& \cdots  \hfill
&{{{\bf{\Lambda }}_{{\boldsymbol \phi ^0},{\boldsymbol \phi ^{{N_f} - 1}}}}} \hfill
& {{{\bf{\Lambda }}_{{\boldsymbol \phi ^0}, {\tilde {\boldsymbol \alpha}} }}} \hfill \\

\vdots  \hfill &  \vdots  \hfill &  \ddots  \hfill &  \vdots  \hfill &  \vdots  \hfill  \\

{{{\bf{\Lambda}}_{{\boldsymbol \phi^{{N_f} - 1}},\boldsymbol\tau }}} \hfill
& {{{\bf{\Lambda }}_{{\boldsymbol \phi ^{{N_f} - 1}},{\boldsymbol \phi ^0}}}} \hfill&\cdots\hfill
& {{{\bf{\Lambda }}_{{\boldsymbol \phi ^{{N_f} - 1}},{\boldsymbol \phi ^{{N_f} - 1}}}}} \hfill
& {{{\bf{\Lambda }}_{{\boldsymbol \phi ^{{N_f} - 1}},{\tilde {\boldsymbol \alpha}} }}} \hfill  \\

{{{\bf{\Lambda }}_{\tilde {\boldsymbol \alpha}, {\boldsymbol \tau} }}} \hfill
& {{{\bf{\Lambda }}_{\tilde {\boldsymbol \alpha}, {\boldsymbol \phi ^0}}}} \hfill
&  \cdots  \hfill
& {{{\bf{\Lambda}}_{\tilde {\boldsymbol \alpha}, {\boldsymbol \phi ^{{N_f} - 1}}}}} \hfill
& {{{\bf{\Lambda }}_{\tilde {\boldsymbol \alpha}, \tilde{ \boldsymbol \alpha }}}} \hfill  \\
\end{matrix}  \right]\!\!,
\end{equation}
}
where
\begin{equation}                  \label{EqAtautau}
{{\bf{\Lambda }}_{{{\boldsymbol \tau, \boldsymbol \tau}}}}{=}\left[
\begin{matrix}
{{{\bf{\Lambda }}_{{{{\tau_1 }}},{{{\tau_1 }}}}}}
&  \cdots
& {{{\bf{\Lambda }}_{{{{\tau_1}}},{{{\tau_L}}}}}}            \\
\vdots        &  \ddots      &  \vdots                       \\
{{{\bf{\Lambda }}_{{{{\tau_L }}},{{{\tau_1 }}}}}} &
\cdots
& {{{\bf{\Lambda }}_{{{{\tau_L }}},{{{\tau_L}}}}}}           \\
\end{matrix} \right].
\end{equation}

Similarly, the other sub-matrices in \eqref{EqFIMSensingparemeter2} are all $L \times L$ matrices, consistent with the form presented in \eqref{EqAtautau}.

Based on \eqref{Eqelement eta}, the detailed expressions of each element in \eqref{EqFIMSensingparemeter2} are provided as follows:
\begin{equation}    \label{EqTauTau}
\begin{aligned}
& {{\bf{\Lambda }}_{\boldsymbol \tau, \boldsymbol \tau}} \left( {i,j} \right)
=  {\frac{2}{{{\sigma ^2}}}}
{\rm Re}\left\{ {\frac{{\partial {\boldsymbol \mu_{\rm s}^{\text{H}}}}}
{{\partial \tau _i}}  \frac{{\partial {\boldsymbol \mu_{\rm s}} }}{{\partial \tau _j}}} \right \}  \\
& =  \frac{2}{{{\sigma ^2}}}
{{\rm Re}}\left\{ {\sum\limits_{\kappa  = 0}^{{N_{\text f}} - 1}}
{{\tilde \alpha_i} {\boldsymbol d\left( {{\boldsymbol \phi _i}} \right)_\kappa ^{\rm H}}
\frac{{\partial \boldsymbol w{{\left( {\tau _i^{\kappa}} \right)}^{\text{H}}}}}{{\partial \tau _i}}
{\tilde \alpha _j}  {\boldsymbol d\left( {{\boldsymbol \phi _j}} \right)_\kappa}
\frac{{\partial \boldsymbol w\left( {\tau_j^{\kappa}} \right)}}{{\partial \tau _j}}} \right \}     \\
& =  \frac{2}{{{\sigma ^2}}}
{\rm Re}\left \{ {{\tilde \alpha_i^{*}} {\tilde \alpha _j} {N_{\text f}}
\frac{{\partial \boldsymbol w{{\left( {\tau _i} \right)}^{\text{H}}}}}{{\partial \tau _i}}
\frac{{\partial \boldsymbol w\left( {\tau_j} \right)}}{{\partial \tau _j}}}  \right \},
\end{aligned}
\end{equation}
where ${\boldsymbol d\left( {{\boldsymbol \phi _i}} \right)_\kappa} = {e^{{\rm{j}}{\phi^{\kappa}_i}}}$,  $i=1,2,...,L$, $j=1,2,...,L$.
Based on the above and expression of \eqref{EqWtN}, the following two assumptions hold.
\begin{equation}   \label{EqAssumptionTau}
\begin{aligned}
\frac{{\partial \boldsymbol w{{\left( {\tau _i^{\kappa}} \right)}}}}{{\partial \tau _i}}  = & \frac{{\partial \boldsymbol w{{\left( {\tau _i^{\kappa}} \right)}}}}{{\partial \tau _i^{\kappa}}},        \:
\boldsymbol w{{\left( {\tau _i^0} \right)}} =  \cdots  = \boldsymbol w{{\left( {\tau _i^{{N_{\text f}} - 1}} \right)}}, \\
& \frac{{\partial \boldsymbol {w}{{\left( {\tau _i^0} \right)}}}}{{\partial \tau _i}}  =  \cdots  = \frac{{\partial \boldsymbol w{{\left( {\tau _i^{{N_{\text f}} - 1}} \right)}}}}{{\partial \tau _i}}.
\end{aligned}
\end{equation}
Additionally, the matrix ${{\bf{\Lambda}}_{\boldsymbol \tau, \boldsymbol \tau}}$ exhibits zero entries for off-diagonal positions, indicated by ${{\bf{\Lambda }}_{\boldsymbol \tau, \boldsymbol \tau}}\left( {i,j} \right)=0$, when $i \neq j$.
${\boldsymbol d\left( {{\boldsymbol \phi _i}} \right)_\kappa^{\rm H}} {\boldsymbol d\left( {{\boldsymbol \phi _j}} \right)_\kappa} = 1$, when $i=j$,
and \eqref{EqTauTau} can be further expressed as \eqref{EqTauTauSNR} based on \eqref{EqSNR} and \eqref{EqBandwidth}.
\begin{equation}    \label{EqTauTauSNR}
{{\bf{\Lambda }}_{\boldsymbol \tau, \boldsymbol \tau}}\left( {i,i} \right) =
4{\pi ^2}{B^2}{N_{\text f}}{T_{\text f}}{f_{\rm{s}}} \cdot SNR
\end{equation}

In this paper, we assume that the SNR is uniform across all paths and consistently represented as $SNR$.

Similarly, the other sub-matrixs are given as follows, and they are derived based on the assumption in \eqref{EqAssumptionTau} and $i=j$.
\begin{equation}   \label{EqTauPhi}
\begin{aligned}
{{\bf{\Lambda }}_{\boldsymbol \tau, \boldsymbol \phi^{\kappa}}}&  \left( {i,j} \right)
{ = }{\frac{2}{{{\sigma ^2}}}}  {\rm Re} \left\{ {\frac{{\partial {\boldsymbol \mu_{\rm s}^{\text{H}}}}}{{\partial \tau _i}}
\frac{{\partial {\boldsymbol \mu_{\rm s}} }}{{\partial {\boldsymbol \phi^{\kappa} _j}}}} \right \}  \cr
& { = \!} \frac{2}{{{\sigma ^2}}}
{\rm Re}\! \left\{\! {{\tilde \alpha_i} {\boldsymbol d\left( {{\boldsymbol \phi _i}} \right)_\kappa ^{\rm H}}
\frac{{\partial \boldsymbol w{{\left( {\tau _i^{\kappa}} \right)}^{\text{H}}}}}{{\partial \tau _i}}
{\tilde \alpha _j}  \frac{{\partial \boldsymbol d\left( {\boldsymbol \phi_j} \right)_{\kappa}}}{{\partial \boldsymbol \phi_j^{\kappa}}}}
\boldsymbol w\left( {\tau_j^{\kappa}} \right)  \!\! \right \}  \cr
& { = }\frac{2}{{{\sigma ^2}}}
{\rm Re}\left\{ {\rm j}{{\tilde \alpha_i} {\tilde \alpha _j}
\frac{{\partial \boldsymbol w{{\left( {\tau _i^{\kappa}} \right)}^{\text{H}}}}}{{\partial \tau _i}}
{{\boldsymbol w\left( {\tau_j^{\kappa}} \right)}}} \right \}  \\
& { = } 0,
\end{aligned}
\end{equation}
where ${\boldsymbol d\left( {{\boldsymbol \phi _i}} \right)_\kappa ^{\rm H}}
\frac{{\partial \boldsymbol d\left( {\boldsymbol \phi_j} \right)_{\kappa}}}{{\partial \boldsymbol \phi_j^{\kappa}}} = {\rm j}$ when $i=j$, and roman letter ${\rm j}$ is the imaginary unit. In this paper, the estimation of delay and Doppler is performed separately, with no coupling between them, which
is consistent with ${{\bf{\Lambda }}_{\boldsymbol \tau, \boldsymbol \phi^{\kappa}}}\left( {i,j} \right) = 0$.
\begin{equation}   \label{EqTauAlpha}
\begin{aligned}
{{\bf{\Lambda }}_{\boldsymbol \tau, \tilde {\boldsymbol \alpha}}}& \left( {i,j} \right)
 { = }{\frac{2}{{{\sigma ^2}}}}
{\rm Re}\left \{ {\frac{{\partial {\boldsymbol \mu_{\rm s}^{\text{H}}}}}{{\partial \tau _i}}
        \frac{{\partial {\boldsymbol \mu_{\rm s}} }}{{\partial {\tilde \alpha_j}}}} \right \}  \cr
& { = }\frac{2}{{{\sigma ^2}}}
{\rm Re} \left \{{{\sum\limits_{\kappa  = 0}^{{N_f} - 1}} {\tilde \alpha_i} {\boldsymbol d\left( {{\boldsymbol \phi _i}} \right)_\kappa ^{\rm H}}
\frac{{\partial \boldsymbol w{{\left( {\tau _i^{\kappa}} \right)}^{\text{H}}}}}{{\partial \tau _i}}
\boldsymbol d \left( {\boldsymbol \phi_j} \right)_{\kappa}}  \boldsymbol w\left( {\tau_j^{\kappa}} \right)\right \}  \cr
& { = }\frac{2}{{{\sigma ^2}}}
{\rm Re} \left\{{{\tilde \alpha_i} {N_f} \frac{{\partial \boldsymbol w{{\left( {\tau _i^{0}} \right)}^{\text{H}}}}}{{\partial \tau _i}}
{{\boldsymbol w\left( {\tau_j^{0}} \right)}}}  \right \},
\end{aligned}
\end{equation}
\begin{equation}   \label{EqPhiAlpha}
\begin{aligned}
 {{\bf{\Lambda }}_{\tilde{\boldsymbol \alpha}, \boldsymbol \phi^{\kappa}}}& \left( {i,j} \right)
{ = }{\frac{2}{{{\sigma ^2}}}}
{\rm Re} \left\{ {\frac{{\partial {\boldsymbol \mu_{\rm s}^{\text{H}}}}}
{{\partial \tilde \alpha _i}}  \frac{{\partial {\boldsymbol \mu_{\rm s}} }}{{\partial {\boldsymbol \phi^{\kappa} _j}}}} \right \}  \cr
& { = }\frac{2}{{{\sigma ^2}}}
{\rm Re}\left\{  {{\boldsymbol d\left( {{\boldsymbol \phi _i}} \right)_\kappa ^{\rm H}}  \boldsymbol w{{\left( {\tau _i^{\kappa}} \right)}^{\text{H}}}
{\tilde \alpha _j}  \frac{{\partial \boldsymbol d\left( {\boldsymbol \phi_j} \right)_{\kappa}}}{{\partial \boldsymbol \phi_j^{\kappa}}}}
\boldsymbol w\left( {\tau_j^{\kappa}} \right)  \right \}     \cr
& { = } 0,
\end{aligned}
\end{equation}
\begin{equation}   \label{EqPhiPhi}
\begin{aligned}
{{\bf{\Lambda }}_{{\boldsymbol \phi^{\kappa}}, \boldsymbol \phi^{\kappa}}}&  \left( {i,j} \right)
{ = }{\frac{2}{{{\sigma ^2}}}}
{\rm Re}\left\{ {\frac{{\partial {\boldsymbol \mu_{\rm s}^{\text{H}}}}}{{\partial {\boldsymbol \phi^{\kappa} _i}}}
        \frac{{\partial {\boldsymbol \mu_{\rm s}} }}{{\partial {\boldsymbol \phi^{\kappa} _j}}}} \right \}  \cr
& { = }\frac{2}{{{\sigma ^2}}}
{\rm Re}\left\{{\tilde \alpha_i}{\frac{{\partial \boldsymbol d\left( {\boldsymbol \phi_i} \right)_{\kappa}}}{{\partial \boldsymbol \phi_i^{\kappa}}}
\boldsymbol w{{\left( {\tau _i^{\kappa}} \right)}^{\text{H}}}
{\tilde \alpha _j}  \frac{{\partial \boldsymbol d\left( {\boldsymbol \phi_j} \right)_{\kappa}}}{{\partial \boldsymbol \phi_j^{\kappa}}}}
\boldsymbol w\left( {\tau_j^{\kappa}} \right)  \right \}  \cr
& { = }\frac{2}{{{\sigma ^2}}}
{\rm Re}\left\{{{\tilde \alpha_i}{\tilde \alpha _j} \boldsymbol w{{\left( {\tau _i^{\kappa}} \right)}^{\text{H}}}
{{\boldsymbol w\left( {\tau_j^{\kappa}} \right)}}} \right\},
\end{aligned}
\end{equation}
where $\frac{{\partial \boldsymbol d\left( {\boldsymbol \phi_i} \right)_{\kappa}}}{{\partial \boldsymbol \phi_i^{\kappa}}}
\frac{{\partial \boldsymbol d\left( {\boldsymbol \phi_j} \right)_{\kappa}}}{{\partial \boldsymbol \phi_j^{\kappa}}} = 1$ when $i=j$, and the submatrix specific to the differentiation of the phase ${\boldsymbol \phi ^{\kappa}}$ shows zero entries for off-diagonal positions, as exemplified by
${{{\bf{\Lambda }}_{{\boldsymbol \phi ^{{N_{\text f}}-1}},{\boldsymbol \phi ^0}}}} = 0$.
\vspace{-10pt}
\begin{equation}     \label{EqAlphaAlpha}
\begin{aligned}
{{\bf{\Lambda }}_{\tilde{\boldsymbol \alpha},\tilde{\boldsymbol \alpha}}} & \left( {i,j} \right)
{ = }{\frac{2}{{{\sigma^2}}}}
{\rm Re} \left\{ {\frac{{\partial {\boldsymbol \mu_{\rm s}^{\text{H}}}}}{{\partial \tilde \alpha _i}}
         \frac{{\partial {\boldsymbol \mu_{\rm s}} }}{{\partial {\tilde \alpha _j}}}} \right\}  \cr
& { = }\frac{2}{{{\sigma^2}}}
{\rm Re}\left\{ {\sum\limits_{\kappa  = 0}^{{N_{\text f}} - 1}}
{{\boldsymbol d\left( {{\boldsymbol \phi _i}} \right)_\kappa ^{\rm H}}  \boldsymbol w{{\left( {\tau _i^{\kappa}} \right)}^{\text{H}}}
  \boldsymbol d\left( {\boldsymbol \phi_j} \right)_{\kappa}}             \boldsymbol w\left( {\tau_j^{\kappa}} \right)  \right\}  \cr
& { = }\frac{2}{{{\sigma ^2}}}
{\rm Re}\left\{ {N_{\text f}} {\boldsymbol w{{\left( {\tau _i^{0}} \right)}^{\text{H}}}   {{\boldsymbol w\left({\tau_j^{0}} \right)}}}  \right\}.
\end{aligned}
\end{equation}

In a manner similar to the derivation of \eqref{EqTauTauSNR}, equations \eqref{EqTauAlpha}, \eqref{EqPhiPhi}, and \eqref{EqAlphaAlpha} can also be rewritten terms of the $\it SNR$ and signal effective bandwidth ${\it B}$, as follows:
\begin{align}
&  {{\bf{\Lambda }}_{\boldsymbol \tau, \boldsymbol \alpha}}\left( {i,i} \right) =
4{\pi ^2}{B^2}{N_{\text f}}{T_{\text f}}{f_{\rm{s}}} \cdot SNR  \cdot {\int_0^{T_{\text f}} {s^{(1)}s(t)}dt \over {\tilde \alpha_i}{\int_0^{T_{\text f}} {(s^{(1)})^2}dt}}, \label{EqTauAlphaSNR} \\
&  {{\bf{\Lambda }}_{{\boldsymbol \phi^{\kappa}}, \boldsymbol \phi^{\kappa}}}\left( {i,i} \right) = {T_{\text f}}{f_{\rm{s}}} \cdot SNR,  \label{EqPhiphiSNR} \\
& {{\bf{\Lambda }}_{\boldsymbol \alpha, \boldsymbol \alpha}}\left( {i,i} \right) = {{N_{\text f}}{T_{\text f}}{f_{\rm{s}}} \over {\tilde \alpha_i^2}} \cdot SNR. \label{EqAlphaAlphaSNR}
\end{align}

Then the FIM for the estimated parameter vector $\boldsymbol\theta_{\rm s}$ is expressed as
\begin{equation}    \label{EqAppEFIMSensing}
\begin{aligned}
{{\bf I}_{{\boldsymbol \theta}_{\rm s}}}& = {\bf J}_{{\rm{s}}}^{\rm{T}}{{\bf{I}}_{\boldsymbol \eta_{\rm s}}}{{\bf J}_{{\rm{s}}}}      \\
& = {\left[ {\begin{array}{*{20}{c}}
              {{{\bf{H}}^{\rm{T}}}{{\bf{\Lambda }}_{\boldsymbol \tau, \boldsymbol \tau }}{\bf{H}}}
& \mathfrak{a}{{{\bf{H}}^{\rm{T}}}{{\bf{\Lambda }}_{\boldsymbol \tau, \boldsymbol \phi^{0}}}{\bf{H}}}
&             {{{\bf{H}}^{\rm{T}}}{{\bf{\Lambda }}_{\boldsymbol \tau, {\tilde {\boldsymbol \alpha}}}}}    \\
  \mathfrak{a}{{{\bf{H}}^{\rm{T}}}{{\bf{\Lambda }}_{\boldsymbol \phi^{0},  {\boldsymbol \tau}}}{\bf{H}}}
& \mathfrak{b}{{{\bf{H}}^{\rm{T}}}{{\bf{\Lambda }}_{\boldsymbol \phi^{0},  \boldsymbol \phi^{0}}}{\bf{H}}}
& \mathfrak{a}{{{\bf{H}}^{\rm{T}}}{{\bf{\Lambda }}_{\boldsymbol \phi^{0},  {\tilde {\boldsymbol \alpha}}}}}       \\
  {{\bf{\Lambda }}_{{\tilde {\boldsymbol \alpha}}, \boldsymbol \tau}} {\bf{H}}
& \mathfrak{a}{{\bf{\Lambda }}_{{\tilde {\boldsymbol \alpha}}, \boldsymbol \phi^{0}}}{{\bf{H}}}
& {{\bf{\Lambda }}_{{\tilde {\boldsymbol \alpha}}, {\tilde {\boldsymbol \alpha}}}}
\end{array}} \right]},
\end{aligned}
\end{equation}
where $\mathfrak{a}=\sum\limits_{\kappa  = 0}^{{N_{\text f}} - 1} {-{\rm j}2\pi \kappa {T_{\text f}}}  {=} -{\rm j} \pi {T_{\text f}}{N_{\text f}}\left( {{N_{\text f}} - 1} \right)$,  \\
$\mathfrak{b} = \sum\limits_{\kappa  = 0}^{{N_{\text f}} - 1} {\left( {2\pi \kappa {T_{\text f}}} \right)\left( {2\pi \kappa {T_{\text f}}} \right)}
              = \frac{{2{\pi ^2}{T_{\text f}^2}{N_{\text f}}\left( {{N_{\text f}} - 1} \right)\left( {2{N_{\text f}} - 1} \right)}}{3}$.

The proof of Corollary \ref{CorolEFIMSense} is thus complete.

\vspace{-5pt}

\subsection{Proof of Proposition \ref{PropoEFIMPPM}}   \label{AppendixProEFIMPPM}

Similar to $\eqref{EqFIMSensingparemeter2}$, the FIM of the observation vector $\boldsymbol \eta_{\rm ppm}$ in PPM ISAC system can be expressed as
\begin{equation}    \label{EqFIMSensingparemeterPPM}
{{\bf{I}}_{\boldsymbol \eta_{\rm ppm}}}
 =  \left[
\begin{matrix}
{{{\bf{\Lambda }}_{{\boldsymbol \tau_{\rm ppm}}, {\boldsymbol \tau_{\rm ppm}}}}} \hfill
& {{{\bf{\Lambda }}_{{\boldsymbol \tau_{\rm ppm}}, {\boldsymbol \phi ^0}}}} \hfill
& \cdots  \hfill
& {{{\bf{\Lambda}}_{\boldsymbol \tau_{\rm ppm}, \tilde{\boldsymbol\alpha} }}} \hfill  \\

{{{\bf{\Lambda }}_{{\boldsymbol \phi ^0}, \boldsymbol \tau_{\rm ppm} }}} \hfill
& {{{\bf{\Lambda }}_{{\boldsymbol \phi ^0},{\boldsymbol \phi ^0}}}} \hfill
& \cdots  \hfill
& {{{\bf{\Lambda }}_{{\boldsymbol \phi ^0}, {\tilde {\boldsymbol \alpha}} }}} \hfill \\

\vdots  \hfill &  \vdots  \hfill &   \ddots  \hfill &   \vdots  \hfill   \\

{{{\bf{\Lambda}}_{{\boldsymbol \phi^{{N_{\text f}} - 1}},\boldsymbol\tau_{\rm ppm} }}} \hfill
& {{{\bf{\Lambda }}_{{\boldsymbol \phi ^{{N_{\text f}}-1}},{\boldsymbol \phi ^0}}}} \hfill
&\cdots\hfill
& {{{\bf{\Lambda }}_{{\boldsymbol \phi ^{{N_{\text f}} - 1}},{\tilde {\boldsymbol \alpha}} }}} \hfill  \\

{{{\bf{\Lambda }}_{\tilde {\boldsymbol \alpha}, {\boldsymbol \tau_{\rm ppm}} }}} \hfill
& {{{\bf{\Lambda }}_{\tilde {\boldsymbol \alpha}, {\boldsymbol \phi ^0}}}} \hfill
&  \cdots  \hfill
& {{{\bf{\Lambda }}_{\tilde {\boldsymbol \alpha}, \tilde{ \boldsymbol \alpha }}}} \hfill  \\
\end{matrix}  \right] ,
\end{equation}
where the numerical values of the elements in
${{\bf I}_{{\boldsymbol \eta}_{\rm ppm}}}$ are consistent with those in ${{\bf I}_{\boldsymbol \eta_{\rm s}}}$, such as
\begin{equation}         \label{EqTauTauPPM}
\begin{aligned}
{{\bf{\Lambda }}_{\boldsymbol \tau_{\rm ppm}, \boldsymbol \tau_{\rm ppm}}}\left( {i,j} \right)
& { = }{\frac{2}{{{\sigma ^2}}}}
{\rm Re} \left\{ {\frac{{\partial {\boldsymbol \mu_{\rm ppm}^{\text{H}}}}}{{\partial \tau _i}}
                  \frac{{\partial {\boldsymbol \mu_{\rm ppm}} }}{{\partial \tau _j}}} \right \}   \cr
& { = } {{\bf{\Lambda }}_{\boldsymbol \tau, \boldsymbol \tau}}\left( {i,j} \right).
\end{aligned}
\end{equation}

The Jacobian matrix $\mathbf {J}_{\rm ppm}$ maps ${\boldsymbol {\eta}_{\rm ppm}}$ to $\boldsymbol\theta_{\rm ppm}$  is written as
\begin{equation}        \label{EqJacPPM}
\begin{aligned}
{\bf J_{{\rm ppm}}}  & = {{\partial {{\boldsymbol \eta }_{{\rm ppm}}}} \over {\partial {\boldsymbol \theta_{{\rm ppm}}}}}
&  = {\left[ \begin{array}{*{20}{l}}
   {\bf{H}}         & {\bf{E}}                 & {\bf{0}}                 & {\bf{0}}  \cr
   {\bf{0}}         & {\bf{0}}                 & {{{\bf{L}_{\rm a}}}}           & {\bf{0}}  \cr
   {\bf{0}}         & {\bf{0}}                 & {\bf{0}}                 & {\bf{I}}  \cr
 \end{array} \right]},                    \cr
\end{aligned}
\end{equation}
where
\begin{equation}
\begin{aligned}
{\bf{E}} & = {{{\partial {\boldsymbol \tau _{\rm{ppm}}}} \over {\partial \Delta \tau _{\rm q}}}} = {[ {\begin{array}{*{20}{c}}
 1 &  \ldots   & 1                   \\
\end{array}}]^{\rm T}_{L \times 1}}.
\end{aligned}
\end{equation}

The FIM for the estimated parameters vector $\boldsymbol\theta_{\rm ppm}$ in ({\ref{Eq theta ppm}) can be expressed as \eqref{EqEFIMPPMAp}.
\begin{figure*}[b]
\vspace{-10pt}
\hrulefill
\begin{equation}        \label{EqEFIMPPMAp}
\begin{aligned}
 {{\bf I}_{{\boldsymbol \theta}_{\rm ppm}}} = {\bf J}_{{\rm{ppm}}}^{\rm{T}}{{\bf{I}}_{\boldsymbol \eta_{\rm ppm}}}{{\bf J}_{{\rm{ppm}}}}
  =  {\left[  {\begin{array}{*{20}{c}}
  {{{\bf{H}}^{\rm{T}}}{{\bf{\Lambda }}_{\boldsymbol \tau, \boldsymbol \tau}}{\bf{H}}}
&  {{{\bf{H}}^{\rm{T}}}{{\bf{\Lambda }}_{\boldsymbol \tau, \boldsymbol \tau}}{\bf{E}}}
& \mathfrak{a} {{{\bf{H}}^{\rm{T}}}{{\bf{\Lambda }}_{\boldsymbol \tau, \boldsymbol \phi^{0} }}{\bf{H}}}
& {{{\bf{H}}^{\rm{T}}}{{\bf{\Lambda }}_{\boldsymbol \tau, {\tilde {\boldsymbol \alpha}}}}}    \\
  {{{\bf{E}}^{\rm{T}}}{{{\bf{\Lambda }}_{\boldsymbol \tau,{\boldsymbol \tau}}}}{\bf{H}}}
&  {{{\bf{E}}^{\rm{T}}}{{\bf{\Lambda }}_{\boldsymbol \tau,{\boldsymbol \tau}}}{\bf{E}}}
& \mathfrak{a} {{{\bf{E}}^{\rm{T}}}{{\bf{\Lambda }}_{\boldsymbol \tau, \boldsymbol \phi^{0} }}}
& {{{\bf{E}}^{\rm{T}}}{{\bf{\Lambda }}_{\boldsymbol \tau, {\tilde {\boldsymbol \alpha}}}}}   \\
  \mathfrak{a} {{{\bf{H}}^{\rm{T}}}{{\bf{\Lambda }}_{\boldsymbol \phi^{0}, \boldsymbol \tau}}{\bf{H}}}
& \mathfrak{a} {{{\bf{H}}^{\rm{T}}}{{\bf{\Lambda }}_{\boldsymbol \phi^{0}, \boldsymbol \tau}}{\bf{E}}}
& \mathfrak{b}{{{\bf{H}}^{\rm{T}}}{{\bf{\Lambda }}_{\boldsymbol \phi^{0},  \boldsymbol \phi^{0}}}{\bf{H}}}
& \mathfrak{a} {{{\bf{H}}^{\rm{T}}}{{\bf{\Lambda }}_{\boldsymbol \phi^{0}, {\tilde {\boldsymbol \alpha}}}}} \\
  {{\bf{\Lambda }}_{{\tilde {\boldsymbol \alpha}}, \boldsymbol \tau}} {\bf{H}}
& {{\bf{\Lambda }}_{{\tilde {\boldsymbol \alpha}}, \boldsymbol \tau}} {\bf{E}}
& {{\bf{\Lambda }}_{{\tilde {\boldsymbol \alpha}}, \boldsymbol \phi^{0}}} {\bf{H}}
& {{\bf{\Lambda }}_{{\tilde {\boldsymbol \alpha}}, {\tilde {\boldsymbol \alpha}}}}
\end{array}} \right]}
\end{aligned}
\end{equation}
\vspace{-10pt}
\end{figure*}
The proof of Proposition \ref{PropoEFIMPPM} is thus complete.
\begin{figure*}[b]
\begin{equation}     \label{EqEFIMSenseAppenBPSK}
\begin{aligned}
{{\bf I}_{{\boldsymbol \theta}_{\rm bpsk}}} = {\bf J}_{{\rm{bpsk}}}^{\rm{T}}{{\bf{I}}_{\boldsymbol \eta_{\rm bpsk}}}{{\bf J}_{{\rm{bpsk}}}}
 =  {\left[  {\begin{array}{*{20}{c}}
  {{{\bf{H}}^{\rm{T}}}{{\bf{\Lambda }}_{\boldsymbol \tau, \boldsymbol \tau}}{\bf{H}}}
& \mathfrak{a} {{{\bf{H}}^{\rm{T}}}{{\bf{\Lambda }}_{\boldsymbol \tau, \boldsymbol \phi^{0} }}{\bf{H}}}
& \mathfrak{a} {{{\bf{H}}^{\rm{T}}}{{\bf{\Lambda }}_{\boldsymbol \tau, \boldsymbol \phi^{0} }}{\bf{E}}}
& {{{\bf{H}}^{\rm{T}}}{{\bf{\Lambda }}_{\boldsymbol \tau, {\tilde {\boldsymbol \alpha}}}}}           \\
  \mathfrak{a} {{{\bf{H}}^{\rm{T}}}{{\bf{\Lambda }}_{\boldsymbol \phi^{0}, \boldsymbol \tau}}{\bf{H}}}
& \mathfrak{b} {{{\bf{H}}^{\rm{T}}}{{\bf{\Lambda }}_{\boldsymbol \phi^{0},  \boldsymbol \phi^{0}}}{\bf{H}}}
& \mathfrak{b} {{{\bf{H}}^{\rm{T}}}{{\bf{\Lambda }}_{\boldsymbol \phi^{0}, \boldsymbol \phi^{0}}}{\bf{E}}}
& \mathfrak{a} {{{\bf{H}}^{\rm{T}}}{{\bf{\Lambda }}_{\boldsymbol \phi^{0}, {\tilde {\boldsymbol \alpha}}}}}  \\
  \mathfrak{a} {{\bf{E}}^{\rm{T}}}{{{\bf{\Lambda }}_{\boldsymbol \phi^{0},{\boldsymbol \tau}}}}{\bf{H}}
& \mathfrak{b} {{{\bf{E}}^{\rm{T}}}{{\bf{\Lambda }}_{\boldsymbol \phi^{0}, \boldsymbol \phi^{0} }}}{\bf{H}}
& \mathfrak{b} {{{\bf{E}}^{\rm{T}}}{{\bf{\Lambda }}_{\boldsymbol \phi^{0},{\boldsymbol \phi^{0}}}}{\bf{E}}}
& \mathfrak{a} {{{\bf{E}}^{\rm{T}}}{{\bf{\Lambda }}_{\boldsymbol \phi^{0}, {\tilde {\boldsymbol \alpha}}}}}       \\
  {{\bf{\Lambda }}_{{\tilde {\boldsymbol \alpha}}, \boldsymbol \tau}} {\bf{H}}
& \mathfrak{a} {{\bf{\Lambda }}_{{\tilde {\boldsymbol \alpha}}, \boldsymbol \phi^{0}}} {\bf{H}}
& \mathfrak{a} {{\bf{\Lambda }}_{{\tilde {\boldsymbol \alpha}}, \boldsymbol \phi^{0}}} {\bf{E}}
& {{\bf{\Lambda }}_{{\tilde {\boldsymbol \alpha}}, {\tilde {\boldsymbol \alpha}}}}
\end{array}} \right]} \
\end{aligned}
\end{equation}
\end{figure*}

\subsection{Proof of Proposition \ref{PropoEFIMBPSK}}   \label{AppendixProEFIMBPSK}

Similar to $\eqref{EqFIMSensingparemeter2}$, the FIM of the observation vector $\boldsymbol \eta_{\rm bpsk}$ in BPSK ISAC system can be expressed as
\begin{equation}    \label{EqFIMSensingparemeterBPSK}
{{\bf{I}}_{\boldsymbol \eta_{\rm bpsk}}} =  \frac{2}{\sigma^{2}} \operatorname{Re}
  \left\{\frac{\partial \boldsymbol{\mu}_{\rm bpsk}^{\rm{H}}}{\partial{\boldsymbol \eta_{\rm bpsk}}}
         \frac{\partial \boldsymbol{\mu}_{\rm bpsk}}{\partial{\boldsymbol \eta_{\rm bpsk}}}\right\},
\end{equation}
where the numerical values of the elements in
${{\bf I}_{{\boldsymbol \eta}_{\rm bpsk}}}$ are consistent with those in ${\bf I}_{\boldsymbol \eta_{\rm s}}$, similar to \eqref{EqTauTauPPM}, such as
\begin{equation}         \label{EqTauTaubpsk}
{\rm Re} \left\{ {\frac{{\partial {\boldsymbol \mu_{\rm bpsk}^{\text{H}}}}}{{\partial \tau _i}}
                  \frac{{\partial {\boldsymbol \mu_{\rm bpsk}} }}{{\partial \tau _j}}} \right \}
= {\rm Re} \left\{ {\frac{{\partial {\boldsymbol \mu_{\rm s}^{\text{H}}}}}{{\partial \tau _i}}
                  \frac{{\partial {\boldsymbol \mu_{\rm s}} }}{{\partial \tau _j}}} \right \}   .
\end{equation}

The Jacobian matrix $\mathbf {J}_{\rm bpsk}$ transfer ${\boldsymbol {\eta}_{\rm bpsk}}$ to $\boldsymbol\theta_{\rm bpsk}$  can be expressed as
\begin{equation}        \label{EqJacBPSKAp}
\begin{aligned}
{\bf J_{{\rm bpsk}}}
= {{\partial {{\boldsymbol \eta }_{{\rm bpsk}}}} \over {\partial {\boldsymbol \theta_{{\rm bpsk}}}}}
= {\left[\begin{array}{*{20}{l}}
   {\bf{H}}         & {\bf{0}}                           & {\bf{0}}               & {\bf{0}}  \cr
   {\bf{0}}         & {{{\bf{L}}_{\rm N}}}               & {\bf{E}}_{\rm N}             & {\bf{0}}  \cr
   {\bf{0}}         & {\bf{0}}                           & {\bf{0}}               & {\bf{I}}  \cr
 \end{array} \right]},
\end{aligned}
\end{equation}
where ${\bf E}_{\rm N} = [\: {\bf E}_0, \cdots, {\bf E}_{N_{\text f}-1} \:]^{\rm T}$,
\begin{equation}
\begin{aligned}
{{\bf E}_{N_{\text f}-1}} = {{{\partial {\boldsymbol \phi ^{{N_{\text f}} - 1}}} \over {\partial \varphi_{\rm{bpsk}}}}}
= {2\pi(N_{\text f}-1)T_{\text f}}{\bf E}.
\end{aligned}
\end{equation}

Then the FIM of estimated parameters vector ${\boldsymbol \theta}_{\rm bpsk}$ in \eqref{Eqthetabpsk} can be expressed as \eqref{EqEFIMSenseAppenBPSK}.

The proof of Proposition \ref{PropoEFIMBPSK} is thus complete.

\vspace{-10pt}
\subsection{Proof of Proposition \ref{PropoEFIMPPMpilot}}    \label{AppendixProEFIMPPMPilot}

The FIM of ${\boldsymbol \eta_{\rm ppm,p}}$ could be given by
\begin{equation}    \label{EqFIMPilot}
{{\bf{I}}_{\boldsymbol \eta_{\rm ppm,p}}} =  \frac{2}{\sigma^{2}} \operatorname{Re}
  \left\{\frac{\partial \boldsymbol{\mu}_{\rm ppm}^{\rm{H}}}{\partial{\boldsymbol \eta_{\rm ppm,p}}}
         \frac{\partial \boldsymbol{\mu}_{\rm ppm}}{\partial{\boldsymbol \eta_{\rm ppm,p}}}\right\},
\end{equation}
where the process for determining the elements of matrix ${{\bf I}_{{\boldsymbol \eta}_{\rm ppm,p}}}$ is the same as that used for matrix
${{\bf I}_{\boldsymbol \eta_{\rm s}}}$.
The Jacobian matrix $\mathbf {J}_{\rm ppm,p}$, which maps ${\boldsymbol {\eta}_{\rm ppm,p}}$ to $\boldsymbol\theta_{\rm ppm}$, can be derived
as
\begin{equation}        \label{EqJacPPMpilot}
\begin{aligned}
{\bf J_{{\rm ppm,p}}}  & = {{\partial {{\boldsymbol \eta }_{{\rm ppm,p}}}} \over {\partial {\boldsymbol \theta_{{\rm ppm}}}}}
&  = {\left[ \begin{array}{*{20}{l}}
   {\bf{H}} & {\bf{0}}        & {\bf{0}}        & {\bf{0}}      \cr
   {\bf{H}} & {\bf{E}}        & {\bf{0}}        & {\bf{0}}      \cr
   {\bf{0}} & {\bf{0}}        & {{{\bf{L}}_{\rm PD}}}  & {\bf{0}}      \cr
   {\bf{0}} & {\bf{0}}        & {\bf{0}}        & {\bf{I}}      \cr
   {\bf{0}} & {\bf{0}}        & {\bf{0}}        & {\bf{I}}
 \end{array} \right]},
\end{aligned}
\end{equation}
where ${\bf{L}_{\rm PD}} = [\: {{\bf{L}}_0}, \ldots, {{\bf{L}}_{P+D-1}} \:]^{\rm T}$.

Furthermore, \eqref{EqEFIMPPMPilot} can be derived from \eqref{EqFIMPilot} and \eqref{EqJacPPMpilot}.

The proof of Proposition \ref{PropoEFIMPPMpilot} is thus complete.

\vspace{-10pt}
\subsection{Proof of Proposition \ref{PropoEFIMBPSKpilot}}    \label{AppendixProEFIMBPSKPilot}

The FIM of ${\boldsymbol \eta_{\rm bpsk,p}}$ can be expressed as
\begin{equation}    \label{EqFIMBPSKPilot}
{{\bf{I}}_{\boldsymbol \eta_{\rm bpsk,p}}} = \frac{2}{\sigma^{2}} \operatorname{Re}
  \left\{\frac{\partial \boldsymbol{\mu}_{\rm bpsk}^{\rm{H}}}{\partial{\boldsymbol \eta_{\rm bpsk,p}}}
         \frac{\partial \boldsymbol{\mu}_{\rm bpsk}}{\partial{\boldsymbol \eta_{\rm bpsk,p}}}\right\}.
\end{equation}
Similarly to \eqref{EqTauTauPPM}, the process for determining the elements of matrix ${{\bf{I}}_{\boldsymbol \eta_{\rm bpsk,p}}}$ is the same as that used for matrix ${{\bf I}_{\boldsymbol \eta_{\rm s}}}$.

The Jacobian matrix $\mathbf {J}_{\rm bpsk,p}$, which maps the observation parameters vector ${\boldsymbol {\eta}_{\rm bpsk,p}}$ to interested parameters vector $\boldsymbol\theta_{\rm bpsk}$ in the BPSK ISAC system can be derived as
\begin{equation}        \label{EqJacBPSKPilot}
\begin{aligned}
{\bf J_{{\rm bpsk,p}}}  & = {{\partial {{\boldsymbol \eta }_{{\rm bpsk,p}}}} \over {\partial {\boldsymbol \theta_{{\rm bpsk}}}}}
&  = {\left[ \begin{array}{*{20}{l}}
   {\bf{H}} & {\bf{0}}        & {\bf{0}}        & {\bf{0}}      \\
   {\bf{H}} & {\bf{0}}        & {\bf{0}}        & {\bf{0}}      \\
   {\bf{0}} & {{{\bf{L}}_{\rm P}}}  & {\bf{0}}       & {\bf{0}}       \\
   {\bf{0}} & {{{\bf{L}}_{\rm D}}}  & {\bf{E}_{\rm D}}       & {\bf{0}}       \\
   {\bf{0}} & {\bf{0}}        & {\bf{0}}        & {\bf{I}}          \\
   {\bf{0}} & {\bf{0}}        & {\bf{0}}        & {\bf{I}}
 \end{array} \right]},
\end{aligned}
\end{equation}
where ${\bf{L}_{\rm P}} = [\: {{\bf{L}}_0}, \ldots, {{\bf{L}}_{P-1}} \:]^{\rm T}$,
${\bf{L}_{\rm D}} = [\: {{\bf{L}}_P}, \ldots, {{\bf{L}}_{P+D-1}} \:]^{\rm T}$,
${\bf{E}_{\rm D}} = [\: {\bf{E}}, \ldots, {\bf{E}} \:]^{\rm T}_{DL  \times  1}$.

Furthermore, \eqref{EqEFIMBPSKPilot} can be derived from \eqref{EqFIMBPSKPilot} and \eqref{EqJacBPSKPilot}.

This completes the proof of Proposition \ref{PropoEFIMBPSKpilot}.

\vspace{-10pt}
\subsection{Proof of Proposition \ref{PropoEFIMDiff2}}    \label{AppendixProEFIMDiff2}

The FIM of the observation vector ${\boldsymbol {\eta}_{\rm diff}}$ can be expressed as following matrix partition:
\begin{equation}    \label{EqFIMSensingparemeterDiff1}
{{\bf{I}}_{\boldsymbol \eta_{\rm diff}}}= \frac{2}{\sigma^{2}} \operatorname{Re}
  \left\{\frac{\partial \boldsymbol{\mu}_{\rm ppm}^{\rm{H}}}{\partial{\boldsymbol \eta_{\rm diff}}}
         \frac{\partial \boldsymbol{\mu}_{\rm ppm}}{\partial{\boldsymbol \eta_{\rm diff}}}\right\},
\end{equation}
where the numerical values of the elements in
${{\bf I}_{{\boldsymbol \eta}_{\rm diff}}}$ are consistent with those in ${\bf I}_{\boldsymbol \eta_{\rm s}}$.

The Jacobian matrix maps ${\boldsymbol {\eta}_{\rm diff}}$ to $\boldsymbol \varpi_{\rm diff}$ is given by
\begin{equation}        \label{EqJacDiifer2}
\begin{aligned}
{\bf P_{{\rm diff}}} = {{\partial {{\boldsymbol \eta }_{{\rm diff}}}} \over {\partial {\boldsymbol \varpi_{\rm diff}}}}
\! = \!{\left[ \begin{array}{*{20}{l}}
{\bf{-I}} &   {\bf{0}} & {\bf{0}}       & {\bf{0}}        & {\bf{0}}     &\cdots     & {\bf{0}}   & {\bf{0}}  \cr
{\bf{I}} &   {\bf{I}} & {\bf{-I}}       & {\bf{0}}        & {\bf{0}}     &\cdots     & {\bf{0}}   & {\bf{0}}  \cr
{\bf{0}} &   {\bf{0}} & {\bf{I}}        & {\bf{I}}        & {\bf{-I}}    &\cdots     & {\bf{0}}   & {\bf{0}}  \cr
{\bf{0}} &   {\bf{0}} & {\bf{0}}        & {\bf{0}}        & {\bf{I}}     &\cdots     & {\bf{0}}   & {\bf{0}}  \cr
\vdots   &   \vdots   &  \vdots         &  \vdots         &  \vdots      & \ddots    &  \vdots    &   \vdots  \cr
{\bf{0}} &   {\bf{0}} & {\bf{0}}        & {{{\bf{0}}}}    & {\bf{0}}     &\cdots     & {\bf{I}}   & {\bf{0}}  \cr
{\bf{0}} &   {\bf{0}} & {\bf{0}}        & {{{\bf{0}}}}    & {\bf{0}}     &\cdots     & {\bf{0}}   & {\bf{I}}  \cr
 \end{array} \right]}. \!\!
\end{aligned}
\end{equation}

The Jacobian matrix $\mathbf {J}_{\rm diff}$, which maps ${\boldsymbol \varpi_{\rm diff}}$ to $\boldsymbol\theta_{\rm ppm}$  can be written as
\begin{equation}        \label{EqJacSening}
\begin{aligned}
{\bf J_{{\rm diff}}}  = {{\partial {\boldsymbol \varpi_{\rm diff}}} \over {\partial {\boldsymbol \theta_{{\rm ppm}}}}}
= {\left[ \begin{array}{{ccccc}}
   {\bf{F}} & {\bf{0}}     & {\bf{0}}                   \\
   {\bf{0}}   & {{{\bf{L}}_{\rm N}}}     & {\bf{0}}      \\
   {\bf{0}}   & {\bf{0}}      & {\bf{I}}
 \end{array} \right]},
\end{aligned}
\end{equation}
where ${\bf F} = {\Bigg[ \begin{array}{{ccccc}}
   {\bf{0}} & {\bf{H}}     & \cdots     & {\bf{0}}         & {\bf{H}}     \\
   {\bf{E}} & {\bf{E}}     & \cdots     & {\bf{E}}         & {\bf{E}}     \\
 \end{array} \Bigg]^{\rm T}}.$

It can be observed that both ${\bf P_{{\rm diff}}}$ and ${\bf J_{{\rm diff}}}$ are non-singular matrices, which ensures that the FIM \eqref{EqEFIMDiff1} is also non-singular.

The proof of Proposition \ref{PropoEFIMDiff2} is thus complete.

\vspace{-5pt}

\bibliographystyle{IEEEtran}
\bibliography{Reference}

\end{document}